\documentclass[11pt,letter]{article}

\usepackage[linkcolor=blue,citecolor=blue,colorlinks=true]{hyperref}
\usepackage[margin=1in]{geometry}
\usepackage{palatino}
\usepackage{complexity}
\usepackage{enumerate,cite,url}
\usepackage{graphicx}
\usepackage{amsmath,amssymb,amsthm, amsfonts}
\usepackage{thmtools}
\usepackage{thm-restate}
\usepackage{color}
\usepackage{tikz}
\usepackage{etoolbox,array, blindtext}
\usepackage{booktabs}
\usepackage{wrapfig}
\usepackage{eso-pic}
\usepackage{forest}
\usepackage{setspace}
\usepackage{tabto}
\usepackage{bbold,cleveref,picture}

\usepackage{amsmath,amssymb,amsfonts,amsthm}
\usepackage{color}
\usepackage{fullpage}
\usepackage{caption}
\usepackage{verbatim,graphicx}
\usepackage{complexity}
\usepackage{algorithmic}
\usepackage{algorithm}
\usepackage{subcaption}
\usepackage{tikz}
\usepackage{tikz-qtree}

\usepackage{todonotes}
\newcounter{todocounter}

\theoremstyle{plain}
\newtheorem{theorem}{Theorem}[section]
\newtheorem{lemma}[theorem]{Lemma}
\newtheorem{proposition}[theorem]{Proposition}
\newtheorem{claim}[theorem]{Claim}

\newtheorem{corollary}[theorem]{Corollary}

\newtheorem{pf-idea}[theorem]{Proof Idea}

\theoremstyle{definition}
\newtheorem{definition}[theorem]{Definition}

\newcommand{\twointer}{\textrm{\sc Two-Intersect-Maj}}
\newcommand{\oneinter}{\textrm{\sc One-Intersect-Maj}}

\newcommand{\maj}{\textrm{\sc {MAJ}}}
\newcommand{\mon}{\textrm{\sc {Monotone}}}

\newcommand{\avoid}{\textrm{\sc Avoid}}
\newcommand{\monncavoid}{\textrm{$\mon$-$\NC^0_3$-$\avoid$}}

\newcommand{\majthreeavoid}{\textrm{$\maj_3$-$\avoid$}}
\newcommand{\majkavoid}{\textrm{$\maj_k$-$\avoid$}}
\newcommand{\oneintermajthree}{\textrm{\oneinter$_3$-\avoid}}
\newcommand{\oneintermajk}{\textrm{$\oneinter_k$-$\avoid$}}
\newcommand{\mychi}{{\cal X}\textrm{$_{2\ell}$}}

\usepackage{complexity}
\newcommand{\comm}[1]{\hfill $\triangleright$ #1}

\newcommand{\half}{\frac{1}{2}}
\newcommand{\range}{{\sf Range}}
\newcommand{\F}{\mathcal{F}}

\renewcommand{\H}{\mathcal{H}}
\renewcommand{\F}{\mathcal{F}}
\newcommand{\callo}{\mathcal{O}}

\usepackage{collect}
\def\movetoappendix{1}

\definecollection{appendix}
\makeatletter
\newenvironment{aproof}[2]
  { \@nameuse{collect}{appendix}
  { \subsection{#1} \label{#2} \begin{proof} } {\end{proof}}
  }{\@nameuse{endcollect}}
\makeatother

\makeatletter
\newenvironment{appsection}[2]
  { \@nameuse{collect}{appendix}
  { \subsection{#1} \label{#2} }
  {}
  }{\@nameuse{endcollect}}
\makeatother

\ifthenelse{\equal{\movetoappendix}{0}}{
        \renewenvironment{aproof}[2]{\begin{proof}} {\end{proof} }
        \renewenvironment{appsection}[2]{} {}
}{}

\title{Range Avoidance in Boolean Circuits via Turan-type Bounds}
\author{Neha Kuntewar\footnote{Indian Institute of Technology Madras, Chennai, India.\\ Email:~{\tt neha.kuntewar@gmail.com, jayalal@cse.iitm.ac.in} }\and Jayalal Sarma$^*$}
\date{}
\begin{document}
\maketitle
\begin{abstract}
   Given a circuit $C : \{0,1\}^n \to \{0,1\}^m$ from a circuit class ${\mathcal C}$, with $m > n$, finding a $y \in \{0,1\}^m$ such that $\forall x \in \{0,1\}^n$, $C(x) \ne y$, is the {\em range avoidance problem} (denoted by ${\cal C}$-$\avoid$). Deterministic polynomial time algorithms (even with access to $\NP$ oracles) solving this problem are known to imply explicit constructions of various pseudorandom objects like hard Boolean functions, linear codes, PRGs etc. 

Deterministic polynomial time algorithms are known for $\NC^0_2$-$\avoid$ when $m > n$, and for $\NC^0_3$-$\avoid$ when $m \ge \frac{n^2}{\log n}$, where $\NC^0_k$ is the class of circuits with bounded fan-in which have constant depth and the output depends on at most $k$ of the input bits. On the other hand, it is also known that $\NC^0_3$-$\avoid$ when $m = n+O\left(n^{2/3}\right)$ is at least as hard as explicit construction of rigid matrices. In fact, algorithms for solving range avoidance for even $\NC^0_4$ circuits imply new circuit lower bounds.

In this paper, we propose a new approach to solving the range avoidance problem via hypergraphs. We formulate the problem in terms of Turan-type problems in hypergraphs of the following kind: for a fixed $k$-uniform hypergraph $H$, what is the maximum number of edges that can exist in $H_C$, which does not have a sub-hypergraph isomorphic to $H$? We show the following:
\begin{itemize}
\item We first demonstrate the applicability of this approach by showing alternate proofs of some of the known results for the range avoidance problem using this framework.
\item We then use our approach to show (using several different hypergraph structures for which Turan-type bounds are known in the literature) that there is a constant $c$ such that $\mon$-$\NC^0_3$-$\avoid$ can be solved in deterministic polynomial time when $m 
> cn^2$. 
\item To improve the stretch constraint to linear, more precisely, to $m > n$, we show a new Turan-type theorem for a hypergraph structure (which we call the \textit{loose $\mychi$-cycles}). More specifically, we prove that any connected 3-uniform linear hypergraph with $m>n$ edges must contain a loose $\mychi$ cycle. This may be of independent interest. 
\item Using this, we show that $\mon$-$\NC^0_3$-$\avoid$ can be solved in deterministic polynomial time when $m > n$, thus improving the known bounds of $\NC^0_3$-\avoid\ for the case of monotone circuits.  In contrast, we note that efficient algorithms for solving $\mon$-$\NC^0_6$-$\avoid$, already imply explicit constructions for rigid matrices.
\item We also generalise our argument to solve the special case of range avoidance for $\NC^0_k$ where each output function computed by the circuit is the majority function on its inputs, where $m>n^2$.
\end{itemize}

\end{abstract}
\newpage
\tableofcontents

\section{Introduction}

Let $C : \{0,1\}^n \to \{0,1\}^m$ be a Boolean circuit with $\{\land, \lor, \lnot \}$ gates, with $m > n$. The range of the function represented by the circuit : $\range(C) = \{ C(x) \mid x \in \{0,1\}^n \}$. Clearly, $\exists y \in \{0,1\}^m$ such that $y \notin \range(C)$. The range avoidance problem (denoted by {\avoid}) asks, given a circuit $C$, with $m > n$, find a $y \notin \range(C)$.

The $\avoid$ problem (introduced by \cite{KKMP'21}) has been shown to have connections to some of the central research questions in circuit lower bounds and pseudorandomness. In particular, even an $\FP^{\NP}$ algorithm for $\avoid$ is known to imply new circuit lower bounds~\cite{Kor'22} and new constructions of many other pseudorandom objects~\cite{Kor'22}.

On the algorithms side, $\avoid$ has a trivial $\ZPP^{\NP}$ algorithm. Indeed, given a circuit $C : \{0,1\}^n \to \{0,1\}^m$ with $m > n$, choose $y \in \{0,1\}^m$ and use the $\NP$ oracle to check if $\exists x \in \{0,1\}^n$ such that $C(x) = y$. Since $m > n$, there are at least $\half$ fraction of $y$s which are outside $\range(C)$, and hence the algorithm succeeds with at least $\half$ probability. Designing a deterministic polynomial time algorithm, with access to an $\NP$ oracle ($\FP^{\NP}$ algorithm) to solve $\avoid$ is a central open problem. 
Recently, \cite{CHR24} obtained the first single-valued $\mathsf{FS}_2\mathsf{P}$ algorithm for $\avoid$ which works infinitely often. Subsequently, \cite{Li23} gave an improved $\mathsf{FS}_2\mathsf{P}$ algorithm that works for all $n$, thus establishing explicit functions in $\mathsf{S}_2\mathsf{E}$ requiring maximum circuit complexity. \cite{CHR24} showed an unconditional zero-error pseudodeterministic algorithm with an $\NP$ oracle and one bit of advice that solves $\avoid$ for infinitely many inputs. These results imply pseudo-deterministic constructions for Ramsey graphs, rigid matrices, pseudo-random generators etc.(See~\cite{CHR24}). 
\cite{ILW23} shows that if there is a deterministic polynomial time algorithm for $\avoid$ then either $\NP = \coNP $ or there does not exist JLS secure $iO$. 

Given the central nature of the problem, it is also meaningful to consider simpler versions first: consider $\mathcal{C}$-$\avoid$ to be the restricted version of $\avoid$ where the circuit is guaranteed to be from the class $\mathcal{C}$. \textit{For which classes of circuits $\mathcal{C}$ do we have an efficient (or $\FP^{\NP}$) algorithm for $\mathcal{C}$-$\avoid$?}
On this frontier, Ren, Santhanam and Wang~\cite{RSW'22} showed that an $\FP^{\NP}$ algorithm for $\mathcal{C}$-$\avoid$ implies breakthrough lower bounds even when $\cal{C}$ is restricted to weaker circuit models such as $\AC^0$ and $\NC^1$ circuits. To go down even further, for every constant $k$, consider the restricted class of circuits $\NC^0_k$ where the depth of the circuit is constant and each output bit depends on at most $k$ of the input bits. In a surprising result, Ren, Santhanam and Wang~\cite{RSW'22} showed that an $\FP^\NP$ algorithm for $\NC^0_4$-$\avoid$ implies $\FP^\NP$ algorithms for $\NC^1$-$\avoid$. Additionally, this will also imply new circuit lower bounds - that there is a family of functions in $\E^{\NP}$ that requires circuits of depth at least $\Omega(n^{1-\epsilon})$. 

Complementing the above, \cite{RSW'22} also exhibited $\FP^\NP$ algorithm for $\avoid$, when $\mathcal{C}$ is restricted to De Morgan formulas of size $s$ with $m> n^{\omega(\sqrt s \log s)}$.  At the low-end regime, \cite{GLW22} showed a polynomial time algorithm for $\NC_2^0$ class, where each output bit depends on at most $2$ input bits. They show a general template for obtaining $\FP^{\NP}$ algorithms for restricted classes via hitting set constructions. In particular, they give $\FP^{\NP}$ algorithms for $\NC^0_k$ circuits, de Morgan formulas, CNF or DNF, provided the stretch is large enough. En route, they also give a general method for obtaining the hitting set in polynomial time using the approximation degree of polynomials. 
Complementing this further, \cite{GGNS23} designed
deterministic polynomial time algorithms for all $\NC^0_k$-$
\avoid$ for $m \ge \frac{n^{k-1}}{\log n}$. For $k=3$, which was the frontier beyond~\cite{GLW22}, this requires $m \ge \frac{n^2}{\log n}$. They also showed the reason for the lack of progress for $\NC^0_3$-$\avoid$ by proving that a deterministic polynomial time algorithm for $\NC^0_3$-$\avoid$, where $m = n+O(n^{2/3})$ would imply explicit construction of rigid matrices in deterministic polynomial time. This demonstrates the importance of the stretch function in the context of $\avoid$ problem.

\paragraph{Our Results:}
In this paper, we propose a new approach to the range avoidance problem for $\NC_k^0$ circuits via Turan-type extremal problems in hypergraphs. For an $\NC^0_k$ circuit $C$, for each function $f_i \in C$ where $i \in [m]$, let $I(f_i)$ denote the set of input variables that $f_i$ depends on. Let $H_C$ denote the hypergraph defined as follows: Let $V$ be the set of inputs in $C$. For $1 \le i \le m$, define a hyperedge $e_i$ as $\{x_j \mid j \in [n], x_j \in I(f_i)\}$. Thus the multiset $E=\{\{e_i \mid i \in [m]\}\}$ has exactly $m$ elements, each of size at most $k$. Let $\cal{F}$ be the family of functions that appear in the circuit $C$, and let $\phi: E \to \cal{F}$ be a labelling of the edges of the hypergraph $H$ with the corresponding function. Without loss of generality, by assigning colors to Boolean bits, say red $R$ for $0$, and blue $B$ for $1$, we can interpret these functions as functions from $\{R,B\}^k \to \{R,B\}$, which induces a color to the hyperedge given any $2$-coloring of the vertices of the hypergraph. A $2$-coloring of the hyperedges of $H$ is said to be a $\phi$-\textit{coloring} if there is a vertex coloring which induces this edge coloring via $\phi$. With this notation, we state the following theorem, which essentially follows from the above definition (see section~\ref{sec:avoid-fixedhypergraphs}).

\begin{restatable}{theorem}{introlemmacoloringtoavoid}
\label{introlemma-coloring-to-avoid}
    Let $C:\{0,1\}^n \rightarrow \{0,1\}^m$ be a circuit and $H_C$ be the corresponding hypergraph and let $\phi$ be the labelling function. Suppose $H_C$ contains a sub-hypergraph $H$ and an edge-coloring of $H$ which is not a $\phi$-coloring, where both the subhypergraph and the edge-coloring can be found in polynomial time. Then, there is a polynomial time algorithm to solve $\avoid$ for $C$. 
\end{restatable}

The main idea of the above proof is that it suffices to identify a sub-hypergraph $H$ in $H_C$ such that there is an edge-coloring of $H$ which is not $\phi$-coloring on $H$. If we can ascertain the existence of a copy of $\H$' by extremal hypergraph theory, then it can be leveraged to solve the $\avoid$ problem. In particular, this formulates range avoidance problem in terms of Turan-type extremal problems in hypergraphs of the following kind - \textit{for a fixed family of hypergraphs $\H$, what is the maximum number of edges that can exist in an $n$-vertex hypergraph that avoids any member of the family $H$ as an (induced) subgraph}? This is particularly well-studied for $k$-uniform hypergraphs and is denoted by $\mathsf{ex}_k(n,\H)$. Notice that the family $\H$ may critically depend on the set of functions $\F$, and the mapping $\phi$, and hence on the circuit $C$. A natural question is, is there a family of hypergraphs $\H$ that works for all circuits in the circuit class\footnote{We remark that while it may not be easy to find the hypergraph structure from the given circuit in general, for the important special cases of the problem mentioned in the previous discussion, just the exhaustive search based on the circuit structure will yield efficient algorithms to find the hypergraph structure.} for which we are interested to solve $\avoid$ problem. 

We first demonstrate immediate, simple applications of this framework for designing algorithms for $\avoid$ in restricted settings. As mentioned above, \cite{GLW22} designed a simple iterative algorithm for solving $\avoid$ when the input is restricted to circuits where each output function depends on at most $2$ input bits. We show that the same special case can also be solved using our approach. Thus, as a warm-up, we derive an alternative proof of the following theorem (originally due to \cite{GLW22}) using our framework.

\begin{theorem}
    There is a deterministic polynomial time algorithm for $\NC^0_2$-$\avoid$ when $m > n$. Using the same framework, there is a polynomial time algorithm for solving $\{${\sc And}$_k$, {\sc Or} $_k \}$-$\avoid$, for a fixed $k$.
\end{theorem}   

As a second demonstration of our framework, we use tools from extremal graph theory to provide deterministic polynomial time algorithms for $\avoid$ when $\F$ contains only $\land$ and $\lor$ functions that depend on exactly $k$ input bits. We remark that such powerful tools are not needed to solve this special case, as the iterative algorithmic idea due to \cite{GLW22} can be extended to this case as well. Nevertheless, we argue that it serves as a demonstration of our technique itself. We state both of these applications as the following proposition.

\begin{proposition}
    There is a deterministic polynomial time algorithm for $\NC^0_2$-$\avoid$ when $m > n$. Using the same framework, there is a polynomial time algorithm for solving $\{${\sc And}$_k$, {\sc Or} $_k \}$-$\avoid$, for a fixed $k$.
\end{proposition}    

We now demonstrate the main application of this framework. Towards describing the setting of our application, we study the complexity of $\avoid$ in the restricted case of when the circuit is monotone. A first observation is that by applying De Morgan's law we can reduce the $\avoid$ (when $m > 2n$) in polynomial time to ${\cal C}$-$\avoid$ where ${\cal C}$ is restricted to monotone circuits, where reduction preserves the depth of the circuit and at most doubles the size of the circuit. Hence, solving $\avoid$ even for monotone circuits implies breakthrough circuit lower bounds. We provide the details of the following proposition in the Appendix~\ref{app:thm-mon-redn}.
\begin{proposition}
\label{thm-mon-redn}
If $m > 2n$, $\avoid$ reduces to $\mon$-$\avoid$ in polynomial time.
\end{proposition}
\begin{aproof}{Proof of \cref{thm-mon-redn}}{app:thm-mon-redn}
Let $C:\{0,1\}^n \rightarrow \{0,1\}^m$ a multi-output circuit with $m>2n$ which is an instance of $\avoid$. We describe how to obtain the circuit $C':\{0,1\}^{2n}\rightarrow \{0,1\}^m$ from $C$. By applying De Morgan's law to push down the negation gates (with appropriate duplication of each gate), we can construct a circuit $D$ equivalent to circuit $C$, with all the negations at the leaves. The circuit $D$ has input literals $\{x_1,\ldots x_n, \neg x_1, \ldots \neg x_n\}$. Now obtain the circuit $C'$ by replacing each input variable $\neg x_i$ by a new variable $x_i'$. Observe that $\range(C) \subseteq \range(C')$. Indeed, consider an arbitrary $y \in \range(C)$ with $a=(a_1, \ldots a_n) \in \{0,1\}^n$ such that $C(a)=y$. Let $a'=(a_1,\ldots a_n, \overline{a_1}, \ldots \overline{a_n}) \in \{0,1\}^{2n}$. Observe that $C'(a')=y$.   
    Since $m > 2n$, $C'$ is a valid input instance for $\mon$-$\avoid$. Hence, it suffices to solve the range avoidance problem for monotone circuits with $m>2n$.
\end{aproof}

Therefore, we can restrict our attention to solving the range avoidance problem for the monotone circuits. We also observe the following corollary for the case of $\NC^0_k$ circuits, which are the current boundary for polynomial time algorithms solving $\avoid$.
    For any $k> 0$, if $m>2n$, $\NC^0_k$-$\avoid$ reduces to $\mon$-$\NC^0_{2k}$-$\avoid$ in deterministic polynomial time.
In particular, when $m> 2n$, $\NC^0_3$-$\avoid$ reduces to $\mon$-$\NC^0_{6}$-$\avoid$ in deterministic polynomial time. In the remaining part of the paper, we use the framework in \cref{introlemma-coloring-to-avoid} to show polynomial time algorithms for $\mon$-$\NC^0_{3}$-$\avoid$ and related problems.

\paragraph{Deterministic Polynomial Time Algorithm for $\monncavoid$ with Quadratic stretch:}

We now show the main technical application of the framework in the case when $\F$ contains only $\maj$ function on $k$ inputs, with the additional constraint that two functions should depend on at most one common input variable. As per the above formulation, this makes the hypergraph to be linear and $k$-uniform. In the setting of $k$-uniform linear hypergraphs, using Turan-type results for specifically designed graphs in $\H$, we can use \cref{introlemma-coloring-to-avoid} for deriving polynomial time algorithms for solving $\avoid$ for specific classes of circuits. 

Turan-type extremal problems for hypergraphs were introduced by \cite{BES73} and bounds are known (see \cite{Kee11} for a survey) for $\mathsf{ex}_k(n,\H)$ for a few hypergraphs $\H$. Bounds are known for when $\H$ is a $k \times k$ grid~\cite{FR13}, wickets~\cite{Sol23}, Fano plane~\cite{KS04} (see section~\ref{sec:prelims} for a brief overview of Turan-type extremal problems, and exact definition of these hypergraphs). 
We exhibit several different fixed hypergraphs $H$ and use them to provide different proofs of the following algorithmic upper bound 

\begin{restatable}{theorem}{thmmonavoid}
\label{thm:monavoid}
$\monncavoid$ when $m > cn^2$ for any constant $c$ can be solved in deterministic polynomial time. 
\end{restatable}

Our proof relies on a reduction of the problem to a more restricted case of $\monncavoid$ where each individual output function is the majority of three input bits. We call this version $\majthreeavoid$. Notice that for any two functions, there can be at most two input bits that both of them can depend on. By using a combinatorial argument on the circuit, we show how to reduce $\majthreeavoid$ to the case where two functions can depend on at most one common input. We denote this version as $\oneintermajthree$. As mentioned above, using \cref{introlemma-coloring-to-avoid} along with the fact that $H_C$ is $k$-uniform, and $\F = \{ \maj_3 \}$, we design polynomial time algorithms for $\oneintermajthree$ where each function is majority on three inputs and two different functions depend on at most one common input bit. We exhibit several different hypergraphs $H$ - wickets (\cref{lem-wicket}), $k$-cage (\cref{lem:alternative1-kcage}), weak Fano plane (\cref{lem:alternative2-weakfanoplane}), $k \times k$ grid (\cref{lem-grid-comb}), $(k,\ell)$-butterfly (\cref{lem:alternative3-butterfly}), $(k,\ell)$-odd kite (\cref{lem:alternative4-odd-kite}) which can be also used to derive polynomial time algorithms for $\oneintermajthree$. Extremal bounds are known for $\mathsf{ex}_k(n,\H)$ only when $\H$ is a weak Fano plane, $3 \times 3$ grid and the wicket. The best known bounds for $\mathsf{ex}_k(n,\H)$ among these are when $\H$ is the hypergraph called a wicket (see \cref{sec:prelims} for a definition), and hence we use it in the proof of \cref{thm:monavoid}.

Observing that the above reduction to the monotone case also doubles the number of input bits on which each function depends, for $\NC^0_k$-$\avoid$ with $m > 2n$, this implies a reduction to $\textrm{\sf Mon}$-$\NC^0_{2k}$-$\avoid$. We use \cref{introlemma-coloring-to-avoid}, with $\H$ as the $k \times k$ grid (see (\cref{lem-grid-comb})), we show that:
\begin{theorem}
There is a deterministic polynomial time algorithm for $\oneintermajk$ when $m > n^2$.
\end{theorem}
However, unlike the case when $k=3$, it is unclear how to reduce $\mon$-$\NC^0_k$-$\avoid$ to $\majkavoid$, and then further to $\oneintermajk$. Indeed, designing polynomial time algorithms for $\mon$-$\NC^0_6$-$\avoid$ itself with $m=n+O(n^{2/3})$ already leads to explicit construction of rigid matrices\cite{GGNS23}, which is an important problem in the area.

\paragraph{Deterministic Polynomial Time Algorithm for $\monncavoid$ with Linear Stretch:}
Deterministic polynomial time algorithms are known\cite{GGNS23} for $\NC^0_3$-$\avoid$ when $m \ge \frac{n^2}{\log n}$ and as mentioned above, improving the stretch constraint to $m = n+O(n^{2/3})$ would imply explicit construction of rigid matrices. We next aim to improve the stretch requirement in the above theorem to linear (in fact, to just $n$), thus improving the known bounds for the case of $\monncavoid$.

Towards this, notice that the above argument for \cref{thm:monavoid} uses the bounds for the Turan number from the literature in a black-box manner. In fact, the quadratic constraints on the stretch function $m$ that we have imposed in \cref{thm:monavoid} can be relaxed by using stricter variants of the \textit{power-bound conjecture} in the context of Turan numbers of linear hypergraphs. In particular, \cite{GL21} conjectures that there exists an $\epsilon$ such that any linear hypergraph on $n$ vertices having more than $\omega(n^{2-\epsilon})$ edges must contain a $(u,u-4)$-hypergraph\footnote{A $(u,u-4)$-hypergraph in this context is a $3$-uniform linear hypergraph which has $u-4$ edges spanning $u$ vertices.} as a subgraph. The specific hypergraph of wicket is an example of a $(9,5)$-hypergraph. However, there are $(9,5)$-hypergraphs which are not suitable for our purpose (See \cref{appsec:95-graphs}). Hence, we need a stronger variant of this conjecture, which insists on having a wicket as a subgraph instead of just $(9,5)$-subgraphs.

Specifically, in the case of wickets, which we critically use in our argument for \cref{thm:monavoid}, \cite{FS24,Sol23} conjectured that the Turan number for $3$-uniform hypergraphs avoiding wickets is $n^{2-o(1)}$, and this will directly improve the constraint on $m$ as $m = n^{2-o(1)}$ for \cref{thm:monavoid} as well. 

To go beyond the limitations posed by the above structures for which Turan number bounds are classically studied, we define a new notion of cycles called $\mychi$-cycles in $3$-uniform linear hypergraphs. We first show a new Turan-type theorem for such cycles for connected hypergraphs, which might be of independent interest.

\begin{restatable}{theorem}{thmextremalchi}\label{thm-loose-chi-bound}
    Any connected 3-uniform linear hypergraph with $m>n$ edges must contain a loose $\mychi$ cycle.
\end{restatable}

In the context of $\oneintermajthree$ where $m>n$, using the above theorem, we show that the corresponding hypergraph $H_C$ contains a loose $\mychi$ cycle. Using the framework of \cref{introlemma-coloring-to-avoid}, where we show (\cref{lem:chi-cycle-not-maj-colorable}) there exists an edge-coloring of $\mychi$ which is not $\maj$-coloring. In addition, we note that although the cycle is not of fixed size, we can find it in $H_C$ in polynomial time. This gives us the following theorem.
\begin{restatable}{theorem}{thmimprovedmonavoid}(\textbf{Main Theorem})
\label{thm:improvedmonavoid}
For $m > n$, $\monncavoid$ can be solved in deterministic polynomial time. 
\end{restatable}

 Thus, while the applicability of \cref{introlemma-coloring-to-avoid} seems to impose constraints such as $k$-uniformity and one-intersection to the cases of $\avoid$ that they can be used to solve, the above theorem indicates that along with other combinatorial reductions to such special cases, it can still lead to useful bounds for the more general problem.

\section{Preliminaries}
\label{sec:prelims}

We study the range avoidance problem for restricted circuit classes. $\mathcal{C}$-$\avoid$ is the following problem: Given a multi-output circuit $C:\{0,1\}^n \rightarrow \{0,1\}^m$ such that $m>n$  where each output function can be computed by a circuit in class $\mathcal{C}$, find a $y\in \{0,1\}^m$ which is outside the range of $C$. 

\paragraph{Hypergraphs:} We collect the preliminaries from the theory of hypergraphs that we use in the paper. We will work with hypergraphs $G(V,E)$ where $E \subseteq 2^{V}$. A hypergraph is linear if two edges intersect in at most one vertex - $\forall e_1, e_2 \in E$, $|e_1 \cap e_2| \le 1$. A hypergraph is said to be $k$-uniform if every edge has exactly $k$ elements from $V$ - $\forall e \in E, |e| = k$. We will be working with $k$-uniform linear hypergraphs, and in particular with $k=3$. We will define some of the hypergraphs that are used in the paper.

A $k \times k$ grid in a hypergraph is a set of $k^2$ vertices $\{v_{11}, v_{12}, \ldots v_{kk}\}$ such that there are edges $r_1, r_2, \ldots r_k, c_1, c_2, \ldots c_k \in E$, corresponding to the vertices in the rows and columns respectively, when the vertices are arranged in the row-major order (See \cref{fig:grid}(a)). A particular special hypergraph (for $k=3$) is called a {\em wicket} is the $3 \times 3$ grid with one edge removed (See \cref{fig:grid}(b)).

\begin{figure}[h!]
    \centering
    \begin{tikzpicture}[scale=0.75]
\draw[-] (1,1) -- (1,4);
\draw[-] (2,1) -- (2,4);
\draw[-] (3,1) -- (3,4);
\draw[-] (4,1) -- (4,4);

\draw[-] (1,1) -- (4,1);
\draw[-] (1,2) -- (4,2);
\draw[-] (1,3) -- (4,3);
\draw[-] (1,4) -- (4,4);

\draw[fill=gray] (1,1) circle (0.1) ;
\draw[fill=gray]  (1,2) circle (0.1) ;
\draw[fill=gray]  (1,3) circle (0.1) ;
\draw[fill=gray]  (1,4) circle (0.1) ;
\draw[fill=gray] (2,1) circle (0.1) ;
\draw[fill=gray]  (2,2) circle (0.1) ;
\draw[fill=gray]  (2,3) circle (0.1) ;
\draw[fill=gray]  (2,4) circle (0.1) ;
\draw[fill=gray] (3,1) circle (0.1) ;
\draw[fill=gray]  (3,2) circle (0.1) ;
\draw[fill=gray]  (3,3) circle (0.1) ;
\draw[fill=gray]  (3,4) circle (0.1) ;
\draw[fill=gray] (4,1) circle (0.1) ;
\draw[fill=gray]  (4,2) circle (0.1) ;
\draw[fill=gray]  (4,3) circle (0.1) ;
\draw[fill=gray]  (4,4) circle (0.1) ;

\node[] at (2.5,0) {(a)};
\end{tikzpicture}
\hspace{2cm}
\begin{tikzpicture}[scale=1]
\draw[-] (1,1) -- (1,3);
\draw[-] (3,1) -- (3,3);
\draw[-] (1,1) -- (3,1);
\draw[-] (1,2) -- (3,2);
\draw[-] (1,3) -- (3,3);

\draw[fill=gray] (1,1) circle (0.07) ;
\draw[fill=gray]  (1,2) circle (0.07) ;
\draw[fill=gray]  (1,3) circle (0.07) ;
\draw[fill=gray] (2,1) circle (0.07) ;
\draw[fill=gray]  (2,2) circle (0.07) ;
\draw[fill=gray]  (2,3) circle (0.07) ;
\draw[fill=gray] (3,1) circle (0.07) ;
\draw[fill=gray]  (3,2) circle (0.07) ;
\draw[fill=gray]  (3,3) circle (0.07) ;

\node[] at (2,0.32) {(b)};
\end{tikzpicture}
\hspace{2cm}
\begin{tikzpicture}[scale=0.35]
\draw[-] (5,9) -- (1.5,3) -- (8.5,3) --cycle;
\draw[-] (5,9) -- (5,3);
\draw[-] (3.25, 6) -- (8.5,3);
\draw[-] (6.75, 6) -- (1.5,3);

\draw[] (5,5) circle (2);
\draw[fill=gray] (5,9) circle (0.18) ;
\draw[fill=gray] (3.25,6) circle (0.18) ;
\draw[fill=gray] (6.75,6) circle (0.18) ;
\draw[fill=gray] (5,5) circle (0.18) ;
\draw[fill=gray] (5,3) circle (0.18) ;
\draw[fill=gray] (1.5,3) circle (0.18) ;
\draw[fill=gray] (8.5,3) circle (0.18) ;

\node[] at (5,1.2) {(c)};
\end{tikzpicture}
    \caption{(a) $4\times 4$ grid (b) A wicket (c) Fano plane}
    \label{fig:grid}
\end{figure}
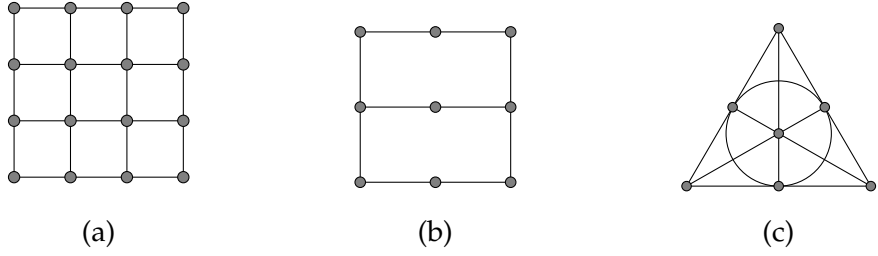

A Berge path is defined as $v_1e_1v_2\ldots v_ke_kv_{k+1}$ where each edge $e_i$ contains vertices $v_i,v_{i+1}$. A Berge path is said to be even if $k$ is even and odd otherwise. We say a hypergraph is connected if there exists a Berge path between any two pairs of vertices.

\paragraph{Turan-type Problems in Hypergraphs:} One of the classical extremal Turan-type problems introduced in the context of hypergraphs by \cite{BES73} is the following: {\em for a fixed set of $k$-uniform hypergraphs $\H$, what is the maximum number of edges that a $k$-uniform linear hypergraph can have, if it does not contain a subgraph isomorphic to any of the hypergraphs in the collection $\H$ of hypergraphs.} This number is denoted by $\mathsf{ex}_k(n,\H)$. The Turan density of $\H$ is denoted by $\pi(\H) = \lim_{n \to \infty} \left( {n \choose k}^{-1}\mathsf{ex}_k(n,\H) \right)$. It is known that a $k$-uniform hypergraph $\H$ has $\pi(\H) = 0$ (also called {\em degenerate}) if and only if it is $k$-partite.

The special $3$-uniform hypergraph called \textit{wicket}, denoted by $W$ (see figure~\ref{fig:grid}(b)), forms an important step in our argument. We will need the following result about 3-uniform linear hypergraphs due to Solymosi~\cite{Sol23}.
\begin{proposition}[\cite{Sol23}]
\label{thm-wicket}
    If a $3$-uniform linear hypergraph does not contain a wicket, then the number of hyperedges is bounded by $o(n^2)$.
\end{proposition}
A slightly weaker result was proven by \cite{GS22} where they showed that $\mathsf{ex}_3(n,{W}) \le \frac{(1-c)n^2}{6}$. Another standard hypergraph that we will be using is the {\em Fano plane}. A Fano plane $F$ is a $3$-uniform linear hypergarph which is isomorphic to the hypergraph $H(V,E)$ with vertex set $V=[7]$ and edge set $E=\{\{1,2,3\},\{3,4,5\},\{1,5,6\},\{3,6,7\},\{2,5,7\},\{1,4,7\},\{2,4,6\}\}$. The following result shows a bound on $\mathsf{ex}_3(n,{F})$.
\begin{proposition}[\cite{KS04}]
If a $3$-uniform linear hypergraph contains more than $ \binom{n}{3}  - \binom{\lfloor n/2 \rfloor}{3} -  {\lceil n/2 \rceil \choose 3}$ then it must contain a Fano plane as a sub-hypergraph.  
\end{proposition}
Coming to $k$-uniform hypergraphs, we use the following result about $k \times k$ grid $G_k$, which exhibits a bound for $\mathsf{ex}_k(n,\{G_k\})$.
\begin{proposition}[\cite{FR13}]\label{thm-grid-turan}
    Let $H$ be a $k$-uniform linear hypergraph with $m>\frac{n(n-1)}{k(k-1)}$ edges. Then $H$ contains a $k \times k$ grid $G_k$ as a sub-hypergraph. 
\end{proposition}

\paragraph{Cycles in Hypergraphs:} Various notions of cycles have been explored in the hypergraph setting. \cite{CCGJ18} consider the notion of linear cycles of length $\ell$ (denoted by $C_\ell$) which is defined as set of  $\ell$ edges such that each pair of adjacent edges $e_i,e_{i+1}$ (modulo $\ell$) intersect in exactly one vertex and each pair of non-adjacent edges are disjoint. \cite{CCGJ18} show that $\mathsf{ex}_k(n,C_{2\ell})\leq c_{k,\ell} n^{1+\frac{1}{\ell}}$ and $\mathsf{ex}_k(n,C_{2\ell +1})\leq c'_{k,\ell} n^{1+\frac{1}{\ell}}$. Another notion of cycles which is well-studied is that of \textit{Berge cycle of length $\ell$}, denoted by ${\cal B}_{\ell}$ which consists of $\ell$ distinct edges such that each hyperedge $e_i$ contains vertices $v_i,v_{i+1}$ where $1\le i <\ell$ and the edge $e_{\ell}$ contains vertices $v_1,v_\ell$ where $v_1,\ldots v_\ell$ are distinct vertices. \cite{GL12} show that $\mathsf{ex}_k(n,{\cal B}_{2\ell})=d_{k,\ell}n^{1+\frac{1}{\ell}}$ and $\mathsf{ex}_k(n,{\cal B}_{2\ell+1})=d'_{k,\ell}n^{1+\frac{1}{\ell}}$. 

\paragraph{$\mychi$-cycles and loose $\mychi$-cycles:} We define a new notion of cycles called $\mychi$-cycles in $3$-uniform linear hypergraphs. We define $\mychi$ to be a relaxation of ${\cal B}_{2\ell}$ where $\exists e_i,e_j$ such that $|e_i \cap e_j|=1$, $i+1 \equiv j \mod 2$ and $|i-j|>2$. We call the $(e_i,e_j)$-pair a $\chi$-structure in the cycle. Alternatively, we can view $\mychi$ using the lens of the Berge path. $\mychi$ can be thought as ${\cal B}_{2\ell_1} \cup {\cal B}_{2\ell_2}\cup \chi$ where ${\cal B}_{2\ell_1} = u_1e_1u_2\ldots u_{2\ell_1}e_{2\ell_1}u_{2\ell_1+1}$, ${\cal B}_{2\ell_2} = v_1e'_1v_2\ldots v_{2\ell_2}e'_{2\ell_2}v_{2\ell_2+1}$, $\chi=(e,e')$ where $e,e' \not \in \{e_1, \ldots e_{2\ell_1}, e'_1,\ldots e'_{2\ell_2}\}$ and $e=\{u_1,2u_{\ell_1},w\}, e'=\{v_1,2v_{\ell_2},w\}$. If we relax the Berge paths to Berge walks (denoted by ${\cal P_{\ell}}$) by allowing repetition of edges and vertices, we get a loose $\mychi$ cycle. In section \ref{sec-mon-nc}, we shall prove extremal bounds on this structure and use it to solve range avoidance for restricted circuit classes.

\begin{figure}
    \centering
      \begin{tikzpicture}[scale=1]

                 \draw[rounded corners, draw=red, fill=red!5!white, rotate around={45:(2,2)}] (1.8, 2.2) rectangle (2.2, -1) {};
                 \draw[rounded corners, draw=red, fill=red!20!white, opacity=0.3, rotate around={-45:(2+2,2)}] (1.8+2, 2.2) rectangle (2.2+2, -1) {};
                 \draw[rounded corners, draw=red, rotate around={-45:(2+2,2)}] (1.8+2, 2.2) rectangle (2.2+2, -1) {};

                \draw[rounded corners, rotate around={30: (2,2)}] (-0.2, 2.2) rectangle (2.2, 1.8) {};
                \draw[rounded corners, rotate around={-30: (2,0)}] (-0.2, 2.2-2) rectangle (2.2, 1.8-2) {};

                \draw[rounded corners] (-0.2+4, 2.2) rectangle (2.2+4, 1.8) {};
                \draw[rounded corners] (-0.2+4, 2.2-2) rectangle (2.2+4, 1.8-2) {};

                
                \draw[rounded corners,rotate around={60:(6,2)}] (-0.2+6, 2.2) rectangle (0.2+6, -0.2) {};
                \draw[rounded corners,rotate around={-60:(6,0)}] (-0.2+6, 2.2) rectangle (0.2+6, -0.2) {};

                    \draw[fill=gray] (2,2) circle (0.1) ;
                    \draw[fill=gray] (0.3,1) circle (0.1) ;
                    \draw[fill=gray] (3,1) circle (0.1) ;
                    \draw[fill=gray] (1+0.25,2-0.45) circle (0.1) ;
                    \draw[fill=gray] (1+0.25,0+0.45) circle (0.1) ;
                    
                    \draw[fill=gray] (2,0) circle (0.1) ;

                    \draw[fill=gray] (4,2) circle (0.1) ;
                    \draw[fill=gray] (5,2) circle (0.1) ;
                    \draw[fill=gray] (6,2) circle (0.1) ;
                    \draw[fill=gray] (7.7,1) circle (0.1) ;
                    \draw[fill=gray] (6,0) circle (0.1) ;
                    \draw[fill=gray] (5,0) circle (0.1) ;\draw[fill=gray] (4,0) circle (0.1) ;
                    \draw[fill=gray] (1+0.9+5,2-0.5) circle (0.1) ;
                    \draw[fill=gray] (1+0.9+5,0+0.5) circle (0.1) ;
                    
            \end{tikzpicture}
    \caption{A $\chi_8$ cycle}
    \label{fig:chi-cycle}
\end{figure}

\begin{appsection}{Special cases of $\mon$-$\avoid$}{sec-special-case}
We now describe some special cases of the monotone range avoidance problem that are simpler than the general case of $\avoid$. 
We consider the depth $1$ sub-circuit obtained by viewing circuits as a composition of two circuits. To denote this, for the rest of this discussion, we assume $D$ is a monotone circuit of depth $d$. For $j \le d$, define $D_j$ to be the sub-circuit of $C$ obtained by removing all the gates and variables below depth $j$, and $\widehat{D_j}$ be the sub-circuit of $C$ obtained by deleting all the gates above depth $j$. Let $m_j$ denote the number of gates at depth $j$ for circuit $C$. Hence $D_j : \{0,1\}^{m_j} \to \{0,1\}^m$ and $\widehat{D_j}: \{0,1\}^n \to \{0,1\}^{m_j}$. By definition, $D = D_j \circ \widehat{D_j}$. We show below that when $j \le 2$, the problem is easier under some constraints for $m_1$ and $m_2$ respectively.
\begin{proposition}
\label{prop:special-case1}
If $m > m_1$ then there is a deterministic polynomial time algorithm that finds a string outside the range of the circuit. 
\end{proposition}
\begin{proof}
Since $m > m_1$, we have that $D_1$ is a valid instance of the $\avoid$ problem. Observe that $D_1$ is a depth $1$ monotone circuit. We claim that $y \not \in \range(D^1) \implies y \not \in \range(D)$. This follows from the fact that if no input setting $x\in \{0,1\}^{m_1}$ in circuit $D^1$ can produce a $y \in \{0,1\}^m$, then no input assignment of $n$ variables in circuit $D$ can ever produce $y$. 

We now show how to find a $y$ outside the range of $D_1$ in deterministic polynomial time. Hence, when $m>m_1$, there is a polynomial time algorithm for $\avoid$.
Similar to the approach in \cite{GLW22}, the idea is that we inductively solve the range avoidance problem for smaller circuits. Consider the first output bit $y_1$ (which is a function $f_1$ of the input variables). We consider the following cases based on the type of functions:
    \begin{description}
    \item{\bf Case 1:} Suppose $f$ is a constant function. Without loss of generality, assume $f$ is the constant zero function. Setting $y_1=1$ and the remaining output bits to arbitrary Boolean values, we obtain a string $y$ which is outside the range of $D_1$.  
    \item{\bf Case 2:} Suppose $f$ is the $\land$ of some variables. In this case, we set output bit $y_1 = 1$ and variables feeding into this gate to be $1$. Thus, we obtain a smaller circuit $D_1':\{0,1\}^{m_j-1}\rightarrow \{0,1\}^{m-1}$. 
    \item{\bf Case 3:} Suppose $f$ is the $\lor$ of some variables. This case can be handled similarly to case $2$. Setting output bit $y_1=0$ and variables feeding into this gate to be $0$, we obtain a smaller circuit $D_1':\{0,1\}^{m_j-1}\rightarrow \{0,1\}^{m-1}$. 
    \end{description}
    After $n$ steps, we obtain a circuit $D_1'':\{0,1\}^0 \to \{0,1\}^{m-m_j}$. The value of the function $f_{m_j+1}$ is fixed by the input assignments. Flipping this value, we obtain a string outside the range of $D_1$. We note that this can be done in $O(m)$ time.
\end{proof}
    Even if $m < m_1$, $m$ is much larger compared to $m_2$, there is an $\FP^\NP$ algorithm for $\mon$-$\avoid$ where $s$ is the size of the circuit.
    \begin{proposition} 
    \label{prop:special-case2}
        If $m > m_2^{\omega(\sqrt{m_2 \log s})}$  there is an $\FP^\NP$ algorithm for $\mon$-$\avoid$ where $s$ is the size of the circuit.
    \end{proposition}
    \begin{proof}
    We show that if $m_2 < \frac{m}{2}$, then the problem can be reduced in polynomial time to {\sc CNF}-$\avoid$. If in addition, $m > m_2^{\omega(\sqrt{m_2 \log s})}$,
    we can use $\FP^{\NP}$ algorithm due to \cite{GLW22} to solve the {\sc CNF}-$\avoid$ instance and this implies the above proposition.
    
    The reduction is as follows. Let $F_1,F_2$ be the set of all $\land$ gates and all $\lor$ gates respectively in the output layer. By pigeon hole principle, we have that either $|F_1| \geq \frac{m}{2}$ or $|F_2| \geq \frac{m}{2}$. Without loss of generality, assume $|F_1|\geq m/2$. Since $m_2 < \frac{m}{2}$, it suffices to solve the problem for $D'$ obtained from $D_2$ by eliminating the top $\lor$ gates. 
    Since the top layer of $D'$ contains only $\land$ gates, we can assume that there are only $\lor$ gates in the second layer of this circuit. Otherwise, we can feed the inputs of this $\land$ gate directly to the $\land$ gate above it. Now $D_2$ is a {\sc CNF} formula.
    \end{proof}
\end{appsection}

\section{Range Avoidance for $\NC^0_k$ via Hypergraphs}
\label{sec:avoid-fixedhypergraphs}

In this section, we describe the main technical tool of the paper by formulating the range avoidance problem as a way of avoiding certain hypergraphs, leading to a Turan-type formulation of the problem.

We recall the notation from the introduction: for an $\NC^0_k$ circuit $C = \left( f_i \right)_{i \in [m]}: \{0,1\}^n \to \{0,1\}^m$, let $\mathcal{H}_C$ denote the hypergraph defined as follows: The set of vertices is $[n]$. For $1 \le i \le m$, define a hyperedge $e_i$ as $\{x_j \mid j \in [n], x_j \in I(f_i)\}$. Thus $E=\{\{e_i \mid i \in [m]\}\}$ has exactly $m$ edges, each of size at most $k$. Let $\cal{F}$ be the family of functions that appear in the circuit $C$, and let $\phi: E \to \cal{F}$ be a labelling of the edges of the hypergraph $H$ with the corresponding function. A $2$-coloring of the hyperedges of $\mathcal{H}$, is said to be a $\phi$-\textit{coloring} if there is a vertex coloring which induces this edge coloring via $\phi$. We argue the following:

\introlemmacoloringtoavoid*
\begin{proof}
    Let $H_C$ be the hypergraph corresponding to the circuit $C$. Let $H$ be a sub-hypergraph of fixed size $\ell$ in $H_C$. We note that since $H$ is of fixed size, we can exhaustively search for a copy of $H$ in $H_C$ in polynomial time. Let $\Gamma: E(H) \rightarrow \{R,B\}$ be an edge-coloring of $H$ which is not a $\phi$-coloring. Let $C'$ be the circuit corresponding to the hypergraph $H$. Let $y=y_1y_2 \ldots y_m$ be defined as follows:
    \[
    y_i=\begin{cases}
        *, &\textrm{ if $e_i \not \in E(H)$}\\
        0, &\textrm{ if $e_i \in E(H)$ and $\Gamma(e_i)=R$}\\
        1, &\textrm{ otherwise}
    \end{cases}
    \]
    where $*$ denotes that $y_i$ can take any value in $\{0,1\}$. We will show that $y \not \in \range(C)$. Let $\Delta:\{R,B\} \rightarrow\{0,1\}$ be a function such that $\Delta(R)=0, \Delta(B)=1$. For the sake of brevity, let $\Delta(x)=\Delta(x_1,\ldots x_n)=\Delta(x_1)\Delta(x_2) \ldots \Delta(x_n)$. It suffices to show the following claim.
    \begin{claim}\label{clm-equivalence}
        $y \in \range(C)$ if and only if $\Delta(y)$ is a valid $\phi$-coloring.
    \end{claim}
    \begin{proof}
            Notice that $y \in \range(C)$ if and only if $\exists x\in \{0,1\}^n \textrm{ such that } C(x)=y$. This is equivalent to $\exists \Pi:V \rightarrow \{R,B\} \textrm{ such that the corresponding } \Gamma:E \rightarrow \{R,B\} \textrm{ is a } \phi\textrm{-coloring}$. Indeed, the latter equivalence can be obtained by setting $\Pi, \Gamma$ such that $\Pi(v_i)=\Delta(x_i)$ and $\Gamma(e_j)=\Delta(y_j)$ for all $i\in [n], j \in [m]$.
    \end{proof}
    Since $H_C$ has an edge coloring $\Gamma$ which is not a $\phi$-coloring, by \cref{clm-equivalence} we obtain a $y\in \{0,1\}^m$ which is outside $\range(C)$.
\end{proof}

\subsection{Warmup : Algorithm for $\NC^0_2$-$\avoid$}
As mentioned in the introduction, \cite{GLW22} described a polynomial time algorithm for solving the range avoidance problem for the $\NC^0_2$ circuit. In the following, we show an alternate proof of the same as an application of the framework described above.

\begin{proposition}
\label{prop:nc02-alternate-proof}
    There is a polynomial time algorithm for $\NC^0_2$-$\avoid$ when $m>n$.
\end{proposition}
\begin{proof}
    Let $C:\{0,1\}^n \rightarrow \{0,1\}^m$ be the given $\NC^0_2$ circuit. \cite{GLW22} observed that there are only four types of $\NC^0_2$ functions possible: {\sc And}, {\sc Or}, {\sc Parity} and constant functions. We can check the type of function in polynomial time. Suppose there is a constant function $f$ in the circuit. Wlog, let this function be a constant zero function. Then, by setting the output of $f$ to $1$ and the other output bits arbitrarily, we get a string that is outside the range of the circuit. 
    
    Now we consider the other case when the circuit contains only $\{\land, \lor, \oplus\}$ gates, where inputs may be negated. We consider the graph $\mathcal{H}_C$ corresponding to the circuit $C$. Since $m>n$, the graph $H_C$ corresponding to the circuit $C$ indeed contains a cycle $Q$. Furthermore, we can find this cycle in polynomial time. Thus, it suffices to obtain an edge-coloring of $Q$ which is not a $\phi$-coloring.
        We consider the following cases based on the function types of the edges participating in the cycle $Q$:
    \begin{description}
        \item[Case 1:] Suppose each edge in $Q$ computes a {\sc Parity} of two inputs which may be negated. Consider an edge $f$ with endpoints $x_i,x_j$. Consider the path $P$ starting at $x_i$ and ending at $x_j$ obtained by deleting $f$ from $Q$. Let $\Gamma$ color all the edges in $Q - f$ to $R$.
        
        Let $x_i=b \in \{R,B\}$. Notice that setting the function to $R$ fixes the color of every other vertex to either $b$ or $\overline{b}$ depending on the function. In particular, it fixes $x_j=b' \in \{b,\overline{b}\}$. This fixes the value of the function computed at $f$, and hence the color of $f$. Thus, flipping the color of $f$ would produce an edge-coloring that is not achievable. Therefore, the following coloring $\Gamma$ is not a $\phi$-coloring.

        \[ \Gamma(e) = \begin{cases}
            R ~~~~~~~~&\textrm{ if $e \in Q - f$}\\
            B ~~~~~~~~&\textrm{ if $e = f$ and $b'=b$}\\
            R ~~~~~~~~&\textrm{ if $e = f$ and $b'=\overline{b}$}\\
        \end{cases}\]
        \item[Case 2:] Suppose there exists an edge $e_1$ in $Q$ which computes {\sc And} or {\sc or} of two inputs which may be negated. Let the cycle $Q$ be $e_1e_2\ldots e_\ell e_1$. Consider the path $e_1e_2\ldots e_{\ell -1}$ obtained by deleting the edge $e_\ell$ from $Q$. Consider the following edge-coloring $\Gamma:Q-e_\ell \rightarrow \{R,B\}$
        \[
            \Gamma(e)= \begin{cases}
                B &\textrm{ if $e \in Q - e_\ell$, $\phi(e)=\land$}\\
                R &\textrm{ if $e \in Q-e_\ell$ and $\phi(e)\in \{\lor, \oplus\}$}\\
            \end{cases}
        \]
        Let $e_1 \cap e_{\ell}=\{x_i\}$ and $e_\ell \cap e_{\ell -1}= \{x_j\}$. We would like to show that by coloring the edges of $Q- e_\ell$ according to $\Gamma$, we end up fixing the colors of $x_i,x_j$ thus fixing the color of $e_{\ell}$ or we obtain an inconsistency. In the former case, flipping the color of $e_\ell$ we obtain an edge-coloring which is not a $\phi$-coloring. 

        It suffices to show that $\Gamma$ fixes the color of $e_\ell$ or we obtain an inconsistent coloring for $Q-e_\ell$. We prove this by induction on the length of the path. For the base case, wlog let $\phi(e_1)=${\sc And}. According to our $\Gamma$, we color $e_1$ to $B$, which fixes the color of both vertices incident on $e_1$. By induction hypothesis, let $\Gamma$ fix the colors of all vertices participating in the path $e_1e_2\ldots e_{k-1}$ for $k \in [\ell]$. In the other case, when the coloring is inconsistent, we are already done. Thus, we would like to argue that $\Gamma$ fixes the other endpoint $y$ of $e_k$ where $y \not \in e_{k-1}$. Notice that, since one input feeding into $e_k$ is already fixed $\Gamma$ gives a way to fix the other input $y$ to $b \in \{R,B\}$ by setting the color of $e_k$ appropriately or we get an inconsistency at this stage. Thus, inductively either we get an inconsistent coloring (in this case, color $e_\ell$ arbitrarily) or we end up fixing the colors of $x_i,x_j$ which in turn fixes the color of the edge $e_\ell$. 

        Hence, flipping the color of $\ell$ fixed by the above assignment produces an edge-coloring which is not a $\phi$-coloring.
        
    \end{description}
    
      By the above case analysis, we have an edge-coloring of $H$ which is not a $\phi$-coloring. By theorem \ref{introlemma-coloring-to-avoid} we obtain a string outside the range of the circuit in polynomial time.
\end{proof} 



\subsection{Warm-up 2: Polynomial time algorithm for $\{\textrm{{\sc And}, {\sc Or}}\}$-$\avoid$}
As our next application, we show a polynomial time algorithm for solving the range avoidance problem when each output function computes an {\sc And} or {\sc Or} function of $k$ input bits for a fixed $k > 0$. Additionally, it satisfies the constraint that any two output functions share at most one input. Thus, the hypergraph $H_C$ corresponding to $C$ is a $k$-uniform linear hypergraph. The sub-hypergraphs we will use in this case are $k$-crown and $C^*$. 
As our next application, we show a polynomial time algorithm for solving the range avoidance problem when each output function computes an {\sc And} or {\sc Or} function of $k$ input bits for a fixed $k$. The sub-hypergraphs we will use in this case are $k$-crown and $C^*$. A $k$-\textit{crown} is a $k$-uniform linear hypergraph consisting of $k$ disjoint hyperedges $\{e_1, \ldots e_k\}$ and an additional hyperedge $e_0$ intersecting each hyperedge $e_1, \ldots e_k$ exactly once. \textit{$C^*$ } is a variant of the $k$-crown where the edges $e_1, \ldots e_{k-2}$ intersect in a common vertex $v \not \in e_0$. We will use the following extremal bound on these subhypergraphs.

\begin{proposition}[\cite{ZBW24}]\label{prop:crown-bound}
    Let $\H$ be a $k$-uniform linear hypergraph with $m>\frac{k(k-2)(n-s)}{k-1}$ where $s$ is the number of vertices with degree at least $(k-1)^2+2$. Then $\H$ contains a $k$-crown or $C^*$.
\end{proposition}

\begin{proposition}
    Let $C:\{0,1\}^n \rightarrow \{0,1\}^m$ be a multi-output circuit where $m>\frac{k(k-2)(n-s)}{k-1}$ and each output function computes {\sc And}$_k$ or {\sc Or}$_k$ for fixed $k$. Additionally, any two output functions share at most one input. Then, there is a polynomial time algorithm for solving $\avoid$ on $C$.
\end{proposition}
\begin{proof}
    We will show an edge coloring of the hypergraphs $k$-crown or $C^*$ which is not a $\phi$-coloring. 
    \[
    \Gamma(e)= \begin{cases}
        B \textrm{ if $e \in \{e_1, \ldots e_k\}, \phi(e) =${\sc And}}\\
        R \textrm{ if $e \in \{e_1, \ldots e_k\}, \phi(e) =${\sc Or}}\\
        B \textrm{ if $e=e_0$, $\exists i \in [k]$ such that $\phi(e_i)=${\sc Or}}\\
        R \textrm{ otherwise}
    \end{cases}
    \] 
    By Proposition \ref{prop:crown-bound} we have that if $m>\frac{k(k-1)(n-s)}{k-1}$ then $\H$ contains a $C^*$ or $k$-crown. Furthermore, we can find these subhypergraphs in polynomial time. By the above argument, we can find an edge coloring of $C^*$ and $k$-crown which is not a $\phi$-coloring. By theorem \ref{introlemma-coloring-to-avoid} we can find a string outside the range of $C$ in polynomial time. 
\end{proof}
We remark that the above special case of $\avoid$ ($\{\textrm{{\sc And}, {\sc Or}}\}$-$\avoid$) can be solved by a direct algorithm similar to the algorithm for $\NC^0_2$-$\avoid$ due to \cite{GLW22}. Nevertheless, the above method is yet another demonstration of our framework for solving the $\avoid$ problem.
\subsection{Polynomial time Algorithm for $\oneinter_3$-$\avoid$}

In this subsection, we show the first application of our formulation of the problem in terms of hypergraphs. We apply it to describe a deterministic polynomial time algorithm for $\oneinter_3$-$\avoid$.

For our purpose, we are interested in a special case of $\phi$-coloring when $\F=\{\maj\}$. We define this formally below:
\begin{definition}[\textbf{$\maj$-coloring}]
     Let $H(V,E)$ be a $k$-uniform hypergraph. We say $\Gamma:E \rightarrow \{R,B\}$ is a $\maj$-coloring if there exists a vertex coloring $\Pi:V \rightarrow \{R,B\}$ such that $\forall e \in E(H)$ we have $\Gamma(e)=B \iff \exists S \subseteq e, |S|\ge \left \lfloor \frac{|V|}{2} \right \rfloor +1$ such that $\forall v\in S, ~\Pi(v)=B$. 
 \end{definition}

We have the following corollary from \cref{introlemma-coloring-to-avoid}.
\begin{corollary}\label{lem-coloring-to-avoid}
    Let $C:\{0,1\}^n \rightarrow \{0,1\}^m$ be a circuit where each output function computes a $\maj$ function. Let $H_C$ be the corresponding hypergraph. Suppose $H_C$ contains a subhypergraph $H$ such that there is an edge-coloring of $H$ which is not $\maj$-coloring. Then there is a polynomial time algorithm to find a string outside the range of $C$. 
\end{corollary}

Thus it is sufficient to exhibit an explicit fixed size linear $3$-uniform hypergraph $H$ which satisfies the conditions of Corollary~\ref{lem-coloring-to-avoid}.

\paragraph{Hypergraph Structure: Wickets -}
A \textit{wicket} is defined as a $3 \times 3$ grid with one edge removed. By \cref{thm-wicket}, if $\H$ contains more than $o(n^2)$ edges, then it contains a wicket. We will show that there is an edge coloring of a wicket which is not a $\maj$-coloring. We note that $\cref{thm-wicket}$ requires the hypergraph to be linear and hence using wickets as subhypergraph would yield an algorithm for $\oneinter_3$-$\avoid$ and not the more general case of $\maj_3$-$\avoid$.

\begin{lemma}\label{lem-wicket}
    Let $\H$ be a wicket. There exists an edge-coloring $\Gamma: E \rightarrow \{R,B\}$ such that it is not a $\maj$-coloring. Furthermore, we can find this in polynomial time.
\end{lemma}
\begin{proof}
    Wlog let the edge $e$ removed from the wicket be a column edge. Let $E=E_1 \cup E_2$ where $E_1$ be the set of three row edges and $E_2$ denote the set of column edges. We will show that the following $\Gamma:E \rightarrow \{R,B\}$ is not a $\maj$-coloring:
    $\Gamma(e)= \begin{cases}
        R &\textrm{ if $e\in E_1$}\\
        B &\textrm{ if $e \in E_2$}
    \end{cases}$.

    Let $U$ be the set of vertices upon which the edges of $E_2$ are incident. For the edges in $E_1$ to be colored $R$, there must be at least three vertices in $U$ that should be colored $R$. Similarly, for the edges in $E_2$ to be $B$, there must be at least $4$ vertices in $U$ that should be colored $B$. Thus, totally there should be at least $7$ distinct vertices in $U$. But $|U|=6$, which is a contradiction. Hence, $\Gamma$ is not a $\maj$-coloring.  
\end{proof}

\noindent The following is an easy corollary of \cref{lem-coloring-to-avoid} and \cref{lem-wicket}. 
\begin{corollary}
\label{cor:majcoloring}
    Let $C:\{0,1\}^n \rightarrow \{0,1\}^m$ be an instance of $\oneinter_3$-$\avoid$ with $m=\Omega(n^2)$. Then we can find a $y \in \{0,1\}^m$ outside the range of $C$ in polynomial time. 
\end{corollary}
\begin{proof}
    Let $H_C$ be the $3$-uniform linear hypergraph corresponding to the circuit $C$. Since $m>o(n^2)$, \cref{thm-wicket} guarantees existence of a wicket in $\H$. 
     Furthermore, we can find such a wicket in polynomial time by simply checking across all subsets of $m$ which are of size $5$. The number of subsets is bounded by ${m \choose 5}$. By \cref{lem-coloring-to-avoid} and \cref{lem-wicket} it follows that we can find a $y$ outside the range of $C$ in polynomial time. 
\end{proof}

\begin{appsection}{Alternative Hypergraphs for Solving $\oneinter_3$-$\avoid$}{subsec:structures}

In this subsection, we show other choices of hypergraphs that can be used in place of wickets in the algorithm above for solving $\oneinter_3$-$\avoid$.

\paragraph{Alternative 1 : $k$-cage:}
We define the graph class as below:\\

\begin{minipage}{0.6\textwidth}
\begin{definition}
    A $k$-cage is a $k$-uniform linear hypergraph $\H(V,E)$ with edge set $E=E_1\cup E_2$ such that $E_1$ consists of pairwsise disjoint edges which intersect exactly at $w\in V$ and $E_2$ consists of $k-1$ edges that are pairwise disjoint over $V \setminus \{w\}$. The picture shown is the $5$-cage.
\end{definition}    
\end{minipage}
\begin{minipage}{0.05\textwidth}
~
\end{minipage}
\begin{minipage}{0.35\textwidth}
    \begin{tikzpicture}[scale=0.6]
        \draw[-] (0,0) -- (-2,0) -- (-4,0) -- (2,0) -- (4,0);
        \draw[-] (-3,1) -- (-1.5,1) -- (0,1) -- (1.5,1) -- (3,1);
        \draw[-] (-2,2) -- (-1,2) -- (0,2) -- (1,2) -- (2,2);
        \draw[-] (-1,3) -- (-0.5,3) -- (0,3) -- (0.5,3) -- (1,3);
        
        \draw[-] (0,4) -- (-1,3) -- (-2,2) -- (-3,1) -- (-4,0);
        \draw[-] (0,4) -- (-0.5,3) -- (-1,2) -- (-1.5,1) -- (-2,0);
        \draw[-] (0,4) -- (0,3) -- (0,2) -- (0,1) -- (0,0);
        \draw[-] (0,4) -- (0.5,3) -- (1,2) -- (1.5,1) -- (2,0);
        \draw[-] (0,4) -- (1,3) -- (2,2) -- (3,1) -- (4,0);
        
        \draw[fill=gray] (0,0) circle (0.1) ;
        \draw[fill=gray] (0,1) circle (0.1) ;
        \draw[fill=gray] (0,2) circle (0.1) ;
        \draw[fill=gray] (0,3) circle (0.1) ;
        \draw[fill=gray] (0,4) circle (0.1) ;
        
        \draw[fill=gray] (0.5,3) circle (0.1) ;
        \draw[fill=gray] (1,2) circle (0.1) ;
        \draw[fill=gray] (1.5,1) circle (0.1) ;
        \draw[fill=gray] (2,0) circle (0.1) ;
        
        \draw[fill=gray] (1,3) circle (0.1) ;
        \draw[fill=gray] (2,2) circle (0.1) ;
        \draw[fill=gray] (3,1) circle (0.1) ;
        \draw[fill=gray] (4,0) circle (0.1) ;
        
        \draw[fill=gray] (-0.5,3) circle (0.1) ;
        \draw[fill=gray] (-1,2) circle (0.1) ;
        \draw[fill=gray] (-1.5,1) circle (0.1) ;
        \draw[fill=gray] (-2,0) circle (0.1) ;
        
        \draw[fill=gray] (-1,3) circle (0.1) ;
        \draw[fill=gray] (-2,2) circle (0.1) ;
        \draw[fill=gray] (-3,1) circle (0.1) ;
        \draw[fill=gray] (-4,0) circle (0.1) ;
    \end{tikzpicture}
\end{minipage}

\vspace{4mm}
\noindent We will show that there exists an edge-coloring $\Gamma$ that is not a $\maj$-coloring.
\begin{lemma}
\label{lem:alternative1-kcage}
    Let $\H$ be a $k$-cage. There exists an edge coloring $\Gamma:E \rightarrow \{R,B\}$ which is not a $\maj$-coloring. Furthermore, we can find such a coloring in polynomial time.
\end{lemma}
\begin{proof}
    Let $\H$ be a $k$-cage with $k$ edge-set $E=E_1 \cup E_2$, where $E_1$ is the set of $k$-edges intersecting in vertex $w$ and $E_2$ is the set of $k-1$ pairwise disjoint edges over $V\setminus \{w\}$. We will show t the following $\Gamma$ is not a $\maj$-coloring:
    \[\Gamma(e)= \begin{cases}
        B &\textrm{  if $e \in E_1$} \\
        R &\textrm{  otherwise} 
    \end{cases}\]
    We consider the following cases based on the parity of $k$:
    \begin{description}
        \item{\bf Case $1$:} Suppose $k$ is odd. Observe that for an edge in $E_2$ to be colored $R$ at least $ \left \lfloor \frac{k}{2} \right \rfloor + 1$ of its vertices should be colored $R$. Since the edges in $E_2$ are pairwise disjoint, we have that there should be at least $k \left( \left\lfloor \frac{k}{2} \right \rfloor + 1 \right)$ vertices in $V\setminus \{w\}$ that are colored $R$. Similarly, for any edge in $e \in E_1$ to be colored $B$, we have that at least $\left \lfloor \frac{k}{2} \right \rfloor $ vertices in $e\setminus \{w\}$ should be colored $B$. Since, any two edges in $E_1$ intersect exactly at vertex $w$, we have that at least $\left( \left \lfloor \frac{k}{2} \right \rfloor \right) \left(k-1 \right)$ vertices in $V\setminus \{w\}$ that should be colored $R$. Thus, totally there must be at least $\left( \left \lfloor \frac{k}{2} \right \rfloor \right) \left(k-1 \right)+\left( \left \lfloor \frac{k}{2} \right \rfloor+1 \right) \left(k\right) = k^2 +k -\left \lfloor k/2 \right \rfloor > k(k-1)$ vertices in $V \setminus \{w\} $. Since, $|V \setminus \{w\}| = k(k-1) $ we get a contradiction.
        
        \item{\bf Case $2$:} Suppose $k$ is even. As earlier for the edges in $E_2$ to be colored $R$ there must be at least $k \left(  \frac{k}{2}  + 1 \right)$ vertices in $V\setminus \{w\}$ that are colored $R$. Since $k$ is odd, for any edge in $e \in E_1$ to be colored $B$, we have that at least $ \frac{k}{2} -1  $ vertices in $e\setminus \{w\}$ should be colored $B$. Hence, at least $\left( \frac{k}{2} -1  \right) \left(k-1 \right)$ vertices in $V\setminus \{w\}$ should be colored $R$. Thus, totally there must be at least $\left(  \frac{k}{2} +1  \right) \left(k \right)+\left( \frac{k}{2} -1  \right) \left(k-1 \right) = k^2 - k/2  -1$ vertices in $V \setminus \{w\} > k(k-1)$. Since, $|V \setminus \{w\}| = k(k-1) $ we get a contradiction.
    \end{description}
     In either case, we show that $\Gamma$ is not a $\maj$-coloring of $\H$.
\end{proof}

\paragraph{Alternative 2 : Weak Fano plane:}
This structure is similar to the Fano plane that was defined in \cref{sec:prelims} with one hyperedge removed.\\

\hspace{-0.6cm}
\begin{minipage}{0.7\textwidth}
\begin{definition}
    A weak Fano plane is a $3$-uniform linear hypergraph which is isomorphic to the hypergraph $\H(V,E)$ with vertex set $V=[7]$ and edge set $$E=\{\{1,2,3\},\{3,4,5\},\{1,5,6\},\{3,6,7\},\{2,5,7\},\{2,4,6\}\}$$
\end{definition}
\end{minipage}
\begin{minipage}{0.05\textwidth}
~
\end{minipage}
\begin{minipage}{0.25\textwidth}
\begin{tikzpicture}[scale=0.4]
\draw[-] (5,9) -- (1.5,3) -- (8.5,3) --cycle;
\draw[-] (3.25, 6) -- (8.5,3);
\draw[-] (6.75, 6) -- (1.5,3);

\draw[] (5,5) circle (2);
\draw[fill=gray] (5,9) circle (0.18) ;
\draw[fill=gray] (3.25,6) circle (0.18) ;
\draw[fill=gray] (6.75,6) circle (0.18) ;
\draw[fill=gray] (5,5) circle (0.18) ;
\draw[fill=gray] (5,3) circle (0.18) ;
\draw[fill=gray] (1.5,3) circle (0.18) ;
\draw[fill=gray] (8.5,3) circle (0.18) ;
\end{tikzpicture}
\end{minipage}

\vspace{4mm}
\noindent We will show that there $\Gamma:E \rightarrow \{R,B\}$ which is not a $\maj$-coloring. 
\begin{lemma}
\label{lem:alternative2-weakfanoplane}
    Let $\H$ be a weak Fano plane hypergraph. Then there is an edge-coloring $\Gamma: E \rightarrow \{R,B\}$ which is not a $\maj$-coloring. Furthermore, we can find this coloring in polynomial time. 
\end{lemma}
\begin{proof}
Let $E_1=\{e|1\in e\}$. We will show that the following $\Gamma$ is not a $\maj$-coloring.  
    \[\Gamma(e)=\begin{cases}
        R &\textrm{  if $e \in E_1$}\\
        B &\textrm{  otherwise}
    \end{cases} \]We consider the following cases:
    \begin{description}
        \item{\bf Case $1$:} Suppose $\Pi(1)=B$. Then for any edge $e\in E_1$ to be colored $R$, it must be that the other two vertices in $e$ should be $R$. In particular, we have $\Pi(3)=\Pi(5)=R$. This implies $\Gamma(\{3,4,5\})=R$ which is a contradiction. 
        \item{\bf Case $2$:} Suppose $\Pi(1)=R$. Then for an edge $e\in E_1$ to be colored $R$ there must be a vertex pair $\{i,j\}\in \{\{2,5\},\{2,6)\},\{3,5\},\{3,6\}\}$ such that $\Pi(i)=\Pi(j)=R$. We will obtain a contradiction in each case: Suppose $\{i,j\}=\{2,5\}$ then $\Gamma(\{2,5,7\})=R$ which is a contradiction. For $\{i,j\}$ being one of $\{2,6\},\{3,5\},\{3,6\}$ we similarly obtain contradictions for the coloring of the edges $\{2,4,6\},\{3,4,6\},\{3,7,6\}$ respectively. 
    \end{description}
    Hence, $\Gamma$ is not a $\maj$-coloring of $\H$.
\end{proof}

\paragraph{Alternative 3 : $(k,\ell)$-odd butterfly:}
A $3$-uniform loose cycle of length $k$, $C:=u_1e_1u_2\ldots u_ke_ku_1$ is a hypergraph over the vertex set $\{u_1, \ldots, u_{k}\}\cup \{v_1,\ldots v_k\}$ such that $u_i,u_{i+1} \in e_i$ and $e_i=(u_i,v_i,u_{i+1})$ where $u_{k+1}=u_1$. 

\hspace{-0.6cm}
\begin{minipage}{0.68\textwidth}
\begin{definition}
    A $(k,\ell)$-odd butterfly is a $3$-uniform hypergraph over $E=\{e_1\ldots e_k\} \cup \{f_1,\ldots f_\ell \}$ consisting of two loose cycles $C_1:=$ $u_1e_1u_2e_2\ldots u_ke_ku_1$ and $C_2:=u_1f_1v_2f_2\ldots v_\ell f_\ell u_1$ where $k,\ell$ are odd. The figure demonstrates a $(3,3)$-butterfly.
\end{definition}
\end{minipage}
\begin{minipage}{0.05\textwidth}
~
\end{minipage}
    \begin{minipage}{.2\textwidth}
        \centering
    \begin{tikzpicture}[scale=0.6]
        \draw[-] (0,0) -- (0,2) -- (0,4);
        \draw[-] (4,0) -- (4,2) -- (4,4);
        \draw[-] (2,2) -- (1,3) -- (0,4);
        \draw[-] (2,2) -- (1,1) -- (0,0);
        \draw[-] (2,2) -- (3,3) -- (4,4);
        \draw[-] (2,2) -- (3,1) -- (4,0);
        
        \draw[fill=gray] (0,0) circle (0.1) ;
        \draw[fill=gray] (0,2) circle (0.1) ;
        \draw[fill=gray] (0,4) circle (0.1) ;
        \draw[fill=gray] (2,2) circle (0.1) ;
        \draw[fill=gray] (1,3) circle (0.1) ;
        \draw[fill=gray] (1,1) circle (0.1) ;
        \draw[fill=gray] (3,1) circle (0.1) ;
        \draw[fill=gray] (4,2) circle (0.1) ;
        \draw[fill=gray] (0+4,0) circle (0.1) ;
        \draw[fill=gray] (0+2,2) circle (0.1) ;
        \draw[fill=gray] (0+4,4) circle (0.1) ;
        \draw[fill=gray] (1+2,3) circle (0.1) ;
        \draw[fill=gray] (1+2,1) circle (0.1) ;
        \end{tikzpicture}
    \end{minipage}

\vspace{3mm}
 \noindent In the following lemma we show an edge coloring function $\Gamma$ which is not a $\maj$-coloring. 
\begin{lemma}
\label{lem:alternative3-butterfly}
    Let $\H$ be a $(k,\ell)$-odd butterfly. Then there is an edge coloring $\Gamma:E(\H) \rightarrow \{R,B\}$ which is not $\maj$-coloring. Furthermore, we can find this coloring in polynomial time.
\end{lemma}
\begin{proof}
    Let $C_1,C_2$ be the two loose cycles of order $k$ in $\H$. Let $C_1:=$ $u_1e_1u_2e_2\ldots u_ke_ku_1$ and $C_2:= u_1f_1v_2f_2\ldots v_\ell f_\ell u_1$. We will show that the following edge-coloring $\Gamma:E \rightarrow \{R,B\}$ is not a $\maj$-coloring.
    \[\Gamma(e_i)=\begin{cases}
        R, &\textrm{     if $i$ is even} \\
        B, &\textrm{      otherwise}
    \end{cases}\]
    \[\Gamma(f_i)=\begin{cases}
        R, &\textrm{     if $i$ is odd} \\
        B, &\textrm{      otherwise}
    \end{cases}\]
    where $i \in [k]$. First, we will show the following claim:
    \begin{claim}\label{clm-color}
        Let $\Gamma: E \rightarrow \{R,B\}$ be a $\maj$-coloring under vertex coloring $\Pi:V\rightarrow \{R,B\}$ such that $\Pi(u_i)=c$ and $\Gamma(e_i)=\overline{c}$  where $c\in \{R,B\}$ and $\overline{c}=\{R,B\} \setminus \{c\}$. Then $\Pi(u_{i+1})=\overline{c}$.
    \end{claim}
    \begin{proof}
        Wlog let $\Pi(u_i)=R$. Since $\Gamma$ is a $\maj$-coloring satisfying $\Gamma(e_i)=B$, the other two vertices in $e_i$ must be $B$. In particular we have that $\Pi(u_{i+1})=B$. 
    \end{proof}
    Suppose $\Pi(u_1)=R$. Then by claim \ref{clm-color} we have that $u_k=R$. We have $\Gamma(e_k)=B$. This implies $\Pi(u_1)=B$ which a contradiction. Similarly, if $\Pi(u_1)=B$ then we can repeatedly apply the claim \ref{clm-color} for $C_2$ to obtain a contradiction. 
\end{proof}

\paragraph{Alternative 4 : $(k,\ell)$-odd kite:}
We define the hypergraph as follows:\\

\hspace{-0.6cm}
\begin{minipage}{0.75\textwidth}
\begin{definition}
    A $(k,\ell)$-odd kite is a $3$-uniform linear hypergraph with edge set $E=\{e_1,\ldots e_k\} \cup \{f_1, \ldots f_k\} \cup \{e'\}$ consisting of two loose odd cycles $C_1=u_1e_1u_2\ldots u_ke_ku_1$, $C_2=u_1f_1v_2f_2\ldots f_{\ell -1}u_kf_\ell u_1$ where $e_k = f_\ell =\{u_0,u_1,u_k\}$, $e'=\{u_0,u_i,v_j\}$ such that $i\in \{2,\ldots k-1\}$, $j \in \{2,3 \ldots \ell -1\} $ and $i \equiv j \mod 2$. The figure shows a $(3,3)$-odd kite.
\end{definition}
\end{minipage}
\begin{minipage}{0.03\textwidth}
~
\end{minipage}
\begin{minipage}{.2\textwidth}
    \begin{tikzpicture}[scale=0.6]
       \draw[-] (0,2) -- (-2,0) -- (-1,1);
        \draw[-] (0,2) -- (2,0) -- (1,1);
        \draw[-] (0,-2) -- (2,0) -- (1,-1);
        \draw[-] (-2,0) -- (0,-2) -- (-1,-1);
        \draw[-] (-2,0) -- (0,0) -- (2,0);
        \draw[-] (0,-2) -- (0,0) -- (0,2);
        
        \draw[fill=gray] (0,0) circle (0.1) ;
        \draw[fill=gray] (0,2) circle (0.1) ;
        \draw[fill=gray] (-2,0) circle (0.1) ;
        \draw[fill=gray] (2,0) circle (0.1) ;
        \draw[fill=gray] (0,-2) circle (0.1) ;
        
        \draw[fill=gray] (1,1) circle (0.1) ;
        \draw[fill=gray] (1,-1) circle (0.1) ;
        \draw[fill=gray] (-1,1) circle (0.1) ;
        \draw[fill=gray] (-1,-1) circle (0.1) ;
        \end{tikzpicture}
    \label{fig:butterly}
    \end{minipage}

\vspace{3mm}
\noindent We will show an edge coloring of $(k,\ell)$ which is not a valid $\maj$-coloring.
\begin{lemma}\label{lem:alternative4-odd-kite}
    Let $\H$ be a $(k,\ell)$-odd kite. Then there exists an edge-coloring $\Gamma:E \rightarrow \{R,B\}$ which is not a valid $\maj$-coloring. Furthermore, we can find such a coloring in polynomial time.
\end{lemma}
\begin{proof}
    We will show that the following edge-coloring $\Gamma$ is not a valid edge-coloring:
    \[
    \Gamma(e)= \begin{cases}
                B, &\textrm{ if $e=e_i\in \{e_1, \ldots e_k \}$ and $i$ is odd} \\
                R, &\textrm{ if $e=e_i\in \{e_1, \ldots e_k \}$ and $i$ is even} \\
                R, &\textrm{ if $e=f_j\in \{f_1, \ldots f_{\ell -1} \}$ and $j$ is odd} \\
                B, &\textrm{ if $e=f_j\in \{f_1, \ldots f_{\ell -1} \}$ and $j$ is even} \\
                R, &\textrm{ if $e=e'=\{u_0,u_i,v_j\}$ and $i$ is even} \\
                B, &\textrm{ if $e=e'=\{u_0,u_i,v_j\}$ and $i$ is odd}
    \end{cases}
    \]
    We consider the following cases based on the colors of vertices $u_1,u_k$:
    \begin{description}
        \item{\bf Case $1$:} Suppose $\Pi(u_1)=\Pi(u_k)=R$. Since two of its vertices are colored $R$, this implies that color of edge $e_k$ is $R$, which is a contradiction.
        \item{\bf Case $2$:} Wlog let $\Pi(u_1)=R, \Pi(u_k)=B$. By repeated application of \cref{clm-color} in cycle $C_1$ we get that $\Pi(u_k)=R$, which is again a contradiction. The other can be handled similarly by applying \cref{clm-color} on cycle $C_2$.
        \item{\bf Case $3$:} Suppose $\Pi(u_1)=\Pi(u_k)=B$. By repeatedly applying \cref{clm-color} in cycles $C_1,C_2$ we get that $\Pi(u_i)= B \iff i \equiv 0 \mod 2$ and $\Pi(v_j)= B \iff j \equiv 0 \mod 2$. This implies that $\Gamma(e')=B \iff i \equiv 0 \mod 2$, which is a contradiction.
    \end{description}
    Since, in either case, we obtain a contradiction, $\Gamma$ is not a valid $\maj$-coloring.
\end{proof}
\end{appsection}
\subsection{Polynomial time Algorithm for $\oneinter_k$-$\avoid$}

In this subsection, we show a second application of the framework proposed in \cref{introlemma-coloring-to-avoid}. Let $C:\{0,1\}^n \rightarrow \{0,1\}^m$ be a multi-output circuit such that each output function computes $\maj$ of $k$ input bits. Let $\H$ be a $k$-uniform hypergraph obtained from $C$ as in the case of \cref{introlemma-coloring-to-avoid}. Since every pair of output functions intersects in at most one input variable, we have that $\H$ is a linear $k$-uniform hypergraph. 

Recall \cref{thm-grid-turan}, which argued that a $k$-uniform hypergraph that has more than $\frac{n(n-1)}{k(k-1)}$ hyperedges must have a $k \times k$ grid contained in it.
We will show that if $\H$ contains a $k \times k$- grid then we can solve range avoidance for the corresponding circuit in polynomial time and hence this will imply an algorithm for $\oneinter_k$-$\avoid$ by using the same framework as \cref{introlemma-coloring-to-avoid}.

\begin{lemma}\label{lem-grid-comb}
    There exists an edge-coloring $\Gamma:E \rightarrow \{R,B\}$ of $k\times k$-grid which is not a $\maj$-coloring. Furthermore, we can find $\Gamma$ in polynomial time.
\end{lemma}
\begin{proof}
    Let $\H(V,E)$ be a $k \times k$ grid. Let $E=E_1 \cup E_2$ where $E_1,E_2$ are the sets of edges corresponding to the rows and columns of the grid, respectively. We define an edge-coloring
    
    $\Gamma(e)= \begin{cases}
        R &\textrm{ if $e \in E_1$}\\
        B &\textrm{ otherwise}
    \end{cases}$
    
    We will show that $\Gamma$ is not a valid $\maj$-coloring. We consider the following cases based on the parity of $k$:
    \begin{description}
        \item[Case 1] Suppose $k$ is odd. For each $e\in E_1$ to be colored $R$, at least $ \left \lfloor \frac{k}{2} \right \rfloor + 1$ of its vertices should be colored $R$. Since the edges in $E_1$ are pairwise disjoint, at least $ k \left ( \left \lfloor \frac{k}{2} \right \rfloor + 1 \right )$ distinct vertices should be colored $R$. Similarly for edges in $E_2$, we have that at least $ k \left ( \left \lfloor \frac{k}{2} \right \rfloor + 1 \right )$ vertices in $V$ should be colored $B$. Hence, there should be at least $ 2k \left (\left \lfloor \frac{k}{2} \right \rfloor + 1 \right)>k^2$ vertices in $V$, which is a contradiction. 
        \item[Case 2] Suppose $k$ is even. By the previous argument, for edges in $E_2$ to get color $B$ at least $k \left (  \frac{k}{2} +1 \right )$ vertices should be colored $B$. For edges in $E_1$ at least $ k \left(  \frac{k}{2} \right ) $ vertices should be colored $R$. So totally, there are at least $ k \left ( k + 1 \right)> k^2$ vertices in $V$, which is again a contradiction. 
    \end{description} 
    In either case, we obtain a contradiction. Hence, $\Gamma$ is not a $\maj$-coloring.
\end{proof}

By proposition \ref{thm-grid-turan} and lemma \ref{lem-grid-comb}, we have the following result:
\begin{theorem}\label{thmoneinter}
    Let $C:\{0,1\}^n \rightarrow \{0,1\}^m$ be an instance of $\oneinter_k$-$\avoid$ and $m>\frac{n(n-1)}{k(k-1)}$. There is a polynomial time algorithm to find a string outside the range of $C$. 
\end{theorem}

\section{Polynomial time Algorithm for $\mon$-$\NC^0_3$-$\avoid$} \label{sec-mon-nc}

In this section, we describe a deterministic polynomial time algorithm for $\mon$-$\NC^0_3$-$\avoid$ for $m=\Omega(n^2)$. The algorithm proceeds in three steps. In the first step, we give a polynomial time reduction from \mon-$\NC^0_3$-{\sc Avoid} to $\maj_3$-{\sc Avoid}. Next, we show a reduction from $\maj_3$-$\avoid$ to $\oneinter_3$-$\avoid$. Finally, we put all the pieces together and describe a polynomial time algorithm for solving $\maj_3$-{\sc Avoid} when $m=\Omega(n^2)$.

\subsection{Reduction from \mon-$\NC^0_3$-{\sc Avoid} to $\maj_3$-{\sc Avoid}:}
We show the following reduction:
\begin{theorem}\label{thm-mon-to-maj}
    There is a polynomial time reduction from $\mon$-$\NC^0_3$-{\sc Avoid} to $\maj_3$-{\sc Avoid}.
\end{theorem}
\begin{proof}
    Let $C:\{0,1\}^n \rightarrow \{0,1\}^m$ be a monotone circuit with $m>n$. We will obtain a circuit $C':\{0,1\}^{n'}\rightarrow \{0,1\}^{m'}$ such that $n'\le n$ ,$m' \le m$ and $n'<m'$, where each output function computes majority of three input bits. Furthermore, we can find a $y=y_1y_2\ldots y_m$ outside the $\range(C)$ from $y' \not \in \range(C')$ in polynomial time by setting the remaining $m-m'$ many bits of $y$ to arbitrary values. We remark that this proof technique is inspired by \cite{GLW22}.

    To begin with, note that there are only a few types of monotone functions that depend on $3$ bits. For each of these cases we will show a reduction from $C:\{0,1\}^n \rightarrow \{0,1\}^m$ to a smaller circuit $C':\{0,1\}^{n-1} \rightarrow \{0,1\}^{m-1}$ such that a string outside the range of $C'$ gives a string which is outside the range of $C$. Let $C_j$ denote the sub-circuit of $C$ that computes the $j$-th bit of the output. Let $W_i:=\{x \in \{0,1\}^3 \mid |x|_1 =i\}$ be the set of inputs with weight $i$, for $i \in \{0,1,2,3\}$. We describe the reduction for each type of function except when the function is the $\maj_3$ function. We iteratively apply this sequence of reduction rules for each $j \in \{1, 2 \ldots m\}$, we will end up with a circuit where each output bit is the $\maj_3$ function in terms of input variables.

    \begin{description}
        \item{\bf Case 1: $C_j^{-1}(0)=\emptyset$:} That is $C_j$ computes a constant $1$ function. Then setting $y_1=0$ and the remaining output bits arbitrarily gives a string outside the range of $C$. 
        
        \item{\bf Case 2: $C_j^{-1}(0)= W_0$:} Suppose $C_j$ computes {\sf OR} of three input variables $x_1, x_2,x_3$. Then setting $y_1=0$ and all the input variables $x_1,x_2,x_3$ to $0$ we obtain a smaller circuit $C':\{0,1\}^{n-3}\rightarrow \{0,1\}^{m-1}$ such that $y' \not \in Range(C')$ then $y=0y'$ is outside the range of $C$.
        \item{\bf Case 3: $C_j^{-1}(0)= W_0 \cup \{\alpha\}$ }where $\alpha \in W_1$: Without loss of generality, we assume $\alpha = 001$. In this case, setting $y_1=0$ and the input variables $x_1,x_2$ to $0$ we obtain a smaller circuit $C':\{0,1\}^{n-2}\rightarrow \{0,1\}^{m-1}$. The same argument applies when $\alpha=100$, and $\alpha=010$.
        \item{\bf Case 4: $C_j^{-1}(0)= W_0 \cup \{\alpha_1, \alpha_2\}$ }where $\alpha_1,\alpha_2 \in W_1$: Without loss of generality, let $\alpha_1 = 001$, and $\alpha_2 = 010$. By the property of monotone functions, we have that $C_j$ evaluates to $1$ on inputs $\{101, 110, 111\}$. Then setting $y_1=0$ and $x_1=0$, yields a circuit $C':\{0,1\}^{n-1} \rightarrow \{0,1\}^{m-1}$. Same argument applies in the case when $\alpha_1 = 010$, $\alpha_2 = 100$ and also in the case when $\alpha_1 = 001, \alpha_2 = 100$.
        \item{\bf Case 5: $C_j^{-1}(0)= W_0 \cup W_1 \cup \{\alpha\}$ }where $\alpha \in W_2$: Without loss of generality, assume $\alpha=011$. Setting $y_1=1$ and $x_1=1$ we get a smaller circuit $C':\{0,1\}^{n-1}\rightarrow \{0,1\}^{m-1}$. The other cases when $\alpha=101$ and $\alpha=110$ can be handled similarly. 
        \item{\bf Case 6: $C_j^{-1}(0)= W_0 \cup W_1 \cup \{\alpha_1, \alpha_2\}$ }where $\alpha_1,\alpha_2 \in W_2$: Without loss of generality, assume $\alpha_1=011, \alpha_2=110$. Setting $y_1=1$ and $x_1=1$ we get a smaller circuit $C':\{0,1\}^{n-1}\rightarrow \{0,1\}^{m-1}$. The same argument applies when $W_2\setminus \{\alpha_1,\alpha_2\}=110$ or $011$.
        \item{\bf Case 7: $C_j^{-1}(0)= W_0 \cup W_1 \cup W_2$:} In this case, $C_j$ computes {\sf AND} of three input variables $x_1, x_2,x_3$. Then setting $y_1=1$ and all the input variables $x_1,x_2,x_3$ to $1$ we obtain a smaller circuit $C':\{0,1\}^{n-3}\rightarrow \{0,1\}^{m-1}$ such that $y' \not \in Range(C')$ then $y=1y'$ is outside the range of $C$. 
        \item{\bf Case 8: $C_j^{-1}(0)= W_0 \cup W_1 \cup W_2 \cup W_3$:} That is, $C_j$ computes a constant zero function. Then setting $y_1=1$ and the remaining output bits arbitrarily gives a string outside the range of $C$. 
    \end{description}
    It can be verified that the only function that is not covered in the above cases is when $C_j$ computes {\sc Maj} on $3$ bits. Note that in each case we eliminate one output bit and at least one input bit. Thus, finally we are left with a circuit $C':\{0,1\}^{n'}\rightarrow \{0,1\}^{m'}$ where $m'>n'$ and each output bit is computed by $\maj_3$ function. 

    We note that we can check the type of function computed by a circuit by simply evaluating the function on all possible input values. Since the function depends on only three variables there are only $8$ input possibilities to check. Hence, overall the reduction can be done in polynomial time.  
\end{proof}

\subsection{Polynomial time Algorithm for $\maj_3$-{\sc Avoid} with Quadratic Stretch}\label{sec:wicket-proof}

We show that there is a deterministic polynomial time algorithm for $\maj_3$-$\avoid$. Recall that $\maj_3$-$\avoid$ instance is a circuit $C : \{0,1\}^n \to \{0,1\}^m$ with constant depth, bounded fan-in and polynomial size, where each output function of the circuit is a majority of exactly three input variables. For each $i \in [m]$, let $f_i$ denote the function corresponding to the $i$-th output bit, and let $I(f_i)$ denote the set of three variables that $f_i$ depends on. We shall demonstrate the working of the algorithm in three steps.

For our purpose, we shall define sub-circuits of $C$ called \textit{clusters} (denoted by $K$) as follows: Let ${\cal O}=\{f_1, \ldots f_m\}$ be the set of output functions of $C$. We define the relation $R: {\cal O} \times {\cal O}$ as follows: We say $(f_i,f_j) \in R$ if $\exists $ some input $x$ that feeds into both $f_i,f_j$ for $i,j \in [m]$. Let $R'$ be the transitive closure of $R$. Observe that $R'$ is an equivalence relation. We define a cluster to be an equivalence class of $R'$.

\paragraph{Step 1: Reduction to a single cluster:}
The following lemma shows that $C$ must contain a cluster with more outputs than inputs and hence we can concentrate on such a cluster.

\begin{lemma}\label{obs-cluster}
    Let $C: \{0,1 \}^n \rightarrow \{0,1\}^m$ be a multi-output circuit with $m>n$, then  there exists a cluster $K$ such that $|K|>|I(K)|$ where $I(K)=\cup_{f\in K} I(f)$.
\end{lemma}
\begin{proof}
    Let $K_1,\ldots K_t$ be the clusters such that $\cup_{i \in [t]} K_i$ covers the set of output functions of $C$. Observe that $\cup_{i \in [t]} I(K_i)$ covers the input set of $C$. By definition, $\forall i,j \in [t]$ such that $i \ne j$ we have that $K_i \cap K_j = \emptyset$ and $I(K_i) \cap I(K_j)=\emptyset$. Assume for the sake of contradiction that for each cluster $K_i$, we have $|K_i|\le |I(K_i)|$. Then by the above observation, we have that $ \left| \cup_{i \in [t]} K_i\right| \le \left| \cup_{i \in [t]} I(K_i) \right|$ implying that $m<n$, which is a contradiction.
\end{proof}

Let $K$ be a cluster guaranteed by proposition \ref{obs-cluster}. Let $C_K$ be the sub-circuit of $C$ corresponding to the cluster $K$. Since, $|K|>|I(K)|$, there must be a string outside the range of $C_K$.  Observe that if $y' \in \{0,1\}^{|K|} \not \in \range (C_K)$ then any $y\in \{0,1\}^m$ which agrees with $y'$ is outside $\range(C)$. Hence, it suffices to solve the problem on $C_K$. 

It is easy to see that if there exist two output functions that intersect in three inputs, then we already have a string outside the range. Hence, we consider the cases when they intersect in fewer than three inputs. Motivated by this, we define the following problems: $\twointer_3$-$\avoid$ ($\oneinter_3$-$\avoid$) is the following problem: Given $C:\{0,1\}^n \rightarrow \{0,1\}^m$ such that output function is $\maj_3$ function and any two functions share at most two (resp. one) inputs, the goal is to find $y \not \in \range(C)$. 

\paragraph{Step 2 : A reduction to $\oneinter_3$-$\avoid$}
For the rest of the proof, $C=C_K$, $m=m'$ and $n=n'$. We show that there is a polynomial time reduction from $\twointer_3$-$\avoid$ problem to $\oneinter_3$-$\avoid$ problem. Finally, in step 3, we will show that there is a deterministic polynomial time algorithm for $\oneinter_3$-$\avoid$ problem. 
\begin{lemma}
\label{thm-2-inter}
There is a polynomial time reduction from $\twointer_3 \textrm{-}\avoid$ to \newline $\oneinter_3 \textrm{-} \avoid$.
\end{lemma}
\begin{proof}
     Let $C:\{0,1\}^n \to \{0,1\}^m$ be a circuit such that each output bit is computed by a  $\maj_3$ function where $m>n$. By step $1$, we know that the set of functions corresponding to $C$ form some cluster $K$. Hence, we obtain $C$ starting with an arbitrary function in $K$ and iteratively adding a new function $f$ to $\callo$, such that $I(f)\cap I(\callo)\neq \emptyset$. Our algorithm will follow this construction procedure, and eliminate functions that appear in that order by setting input variables.
 
     We start with the following claim that helps us eliminate functions from the circuit by setting input variables. Let $f_1, f_2$ be two output functions of $C$ such that $|I(f_1) \cap I(f_2)| = 2$.  Initially ${\cal O} = \{f_1, f_2\}$, and we iteratively apply the following claim for each function $g \in K \setminus {\cal O}$ that is used to build the cluster by the above process. A variable in $I(K)$ is said to be {\em alive} if it is not set to a value in $\{0,1\}$ in the below process.
    \begin{claim}
    \label{claim:2-inter}
    $\forall g \in K\setminus \callo$ such that $|I(g) \cap I({\callo})| = 2$, there is a setting of $g \in \{0,1\}$ and a consistent partial assignment to the input variables (or identification of variables), such that at most one variable in $\callo$ remains alive.
    \end{claim}
    \begin{proof}
    We shall prove this by induction on the construction of the cluster. Initially, let $\callo$ $=\{f_1, f_2\}$, $I(\callo)$ $ = \{x_1, x_2, x_3, x_4\}$. Without loss of generality, let $I(f_1 ) = \{x_1, x_2, x_3\}$ and $I(f_2 ) = \{x_2, x_3, x_4\}$. We observe that if $f_1 = 1$ and $f_2 = 0$, then $x_1 = 1, x_4 =0, x_3 = \neg x_2$ and the number of variables in $I(\callo)$ is exactly $1$. 
    
    By hypothesis, the number of variables in $I(K)$ which are alive after $i$ steps is at most $1$. Now consider a $g$ such that $|I(g) \cap I(K)| = 2$, we would like to show that there is an assignment of values to $g$ (which also sets some input variables) such that in the new cluster, $K \cup \{g\}$, the number of variables that are alive is still at most $1$. The following cases arise depending on the two inputs $g$ share with the cluster. Let $z_1,z_2 \in I(g) \cap I(K)$ and let $y =I(g)\setminus I(K), b \in \{0,1\}$.

    \begin{description}
        \item{\bf Case 1:} \textit{Both $z_1$, $z_2$ are not alive:}
            There are two cases to consider. Suppose $z_1 = z_2 = b$. By setting $g = \overline{b}$, we can obtain a string that is outside the range of $C$.
            The other case is when $z_1 = b$, $z_2 = \overline{b}$. In this case, we set $g=b$ and $y=b$. Therefore, the number of variables in the new cluster is still at most $1$.
        \item{\bf Case 2:} \textit{Exactly one of $z_1$, $z_2$ is alive:}
              Again two cases arise. Consider the first case when $z_1 = b$, $z_2 = x$. Now, we set $g = x= y= \overline{b}$ eliminating all the variables. The other case is: $z_1 = b$, $z_2 = \lnot x$. Similarly, we set $g= \overline{b}$, $x = b$, $y= \overline{b}$ which eliminates all the variables. 
        \item{\bf Case 3:} \textit{Both $z_1$, $z_2$ are alive:} If $z_1,z_2$ that are both alive, then it must be that $z_1=x, z_2=\lnot x$. In this case, we set $g=y=0$ so that the number of variables in the cluster is at most $1$.
    \end{description}
    
    Hence, in all the cases there is an assignment value to $g$ and consistent assignment for variables in $I(g)$ such that there is at most one variable which is alive in $I(\callo)$.
    \end{proof}
    
     Now, we consider the functions $h$ such that $\left| I(h) \cap I(\callo) \right|=1$. Let $w:= I(h) \cap I(\callo)$. If $w$ is assigned to $b\in \{0,1\}$, then we set $h$ to $\overline{b}$ and the input variables in $I(h) \setminus \{w\}$ to $\overline{b}$. Consider the other case when $w$ is alive. Wlog let $w=x$. We set the output of $h$ and its other two inputs to $\overline{x}$. In this case, we have not yet assigned a fixed $0/1$ value to $h$, which we shall fix in the next step. However, note that in either case, the number of alive variables in $\callo$ is at most $1$.

     After at most $n-2$ iterations, we would have fixed each of the inputs to one of $0,1,x$ or $\neg x$. This is because, in the iterative process, each function that we add to the set $\callo$ covers at least one new variable. Now consider an $f\in K\setminus \callo$, which is guaranteed since $m>n$. Observe that $I(f) \subseteq I(\callo)$. Let $I(f) = \{x_1,x_2,x_3\}$. The following two cases arise based on the number of alive variables.

     \begin{description}
        
         \item[Case 1:] Suppose there is no variable in $\callo$ which is alive. Then the value of $f$ is already fixed by the inputs $I(\callo)$. Wlog let this be $b \in \{0,1\}$. Setting the output of $f$ to $\overline{b}$ gives a string outside $\range(C)$. 
         \item[Case 2:] Suppose there is exactly one variable that is alive. Wlog let $x_1=x$.
         \begin{description}
             \item[Case 2(a):] Consider the case when $x_2=x_3=b$ for $b \in \{0,1\}$. Notice that this already forces $f$ to take value $b$. Thus, by setting $f$ to $\overline{b}$ and the functions outputs of functions that were set to $x$ to $b$ we obtain a $y \in \{0,1\}^m$ which is outside $\range(C)$.
             \item[Case 2(b):] Suppose $x_2=b,x_3=\overline{b}$. Then setting the output of $f$ to $b$, we fix the value of $x$ to $b$. Thus, at this stage, there are no variables that are alive in $\callo$. Since, $|\callo|\le n-1$ and $m>n$ there exists a $g \in K\setminus \callo$ such that $I(g) \subseteq I(\callo)$. Since all inputs are fixed, we can now handle this using case 1.       
         \end{description}
         \item[Case 3:] Suppose $f$ has two variables that are alive. Let $x_1=b$ where $b \in \{0,1\}$.
         \begin{description}
             \item[Case 3(a):] Suppose $x_2=x_3=x$. By setting the output of $f$ to $\overline{b}$, we fix the value of $x_2,x_3$ to $\overline{b}$. This eliminates all the variables from $\callo$. Again, this reduces to case 1.
             \item[Case 3(b):] Suppose $x_2=x,x_3=\overline{x}$. This fixes the value of $f$ to $b$. Hence, setting the output of $f$ to $\overline{b}$ gives a string outside the range. 
         \end{description}
         \item[Case 4:] Suppose all the three variables of $f$ are alive. Note that since $|\callo $ $| \le n-1$ and $m>n$, there must be at least two more functions outside $\callo$. If any of these functions satisfy the above cases, then we have already found a solution to the problem. Therefore, we consider the case when all three functions have all three variables that are alive. By PHP, there exist two functions at least two of whose inputs are set to $x (\overline{x})$. This fixes the output of these functions to $x (\overline{x})$. Setting one of the outputs to $1$ and the other to $0$ yields a string outside the range.  
     \end{description}

    Thus, if there exist two functions in $C$ that have at least two common inputs, then as described above, we already have a solution to the range avoidance problem. Otherwise, we have an instance of $\oneinter_3$-$\avoid$. Note that the above steps can be done in polynomial time. Hence, it suffices to show a polynomial time algorithm for this case.
     \end{proof}

Having described all the ingredients of the proof, we are now ready to prove our first algorithm for $\maj_3$-$\avoid$ (which works only for quadratic stretch $m > cn^2$). We apply step $1$ to obtain a sub-circuit $C_K$ whose outputs form a cluster $K$. We note that there must be a cluster $K$ with $n'$ inputs and $m'>cn'^{2}$ outputs. Otherwise, if each cluster with $n_i$ inputs has less than $cn_i^2$ outputs then we get $m=\sum cn_i^2 < cn^2$, which is a contradiction. Next, we observe that any two functions in $C_K$ share at most two input variables since otherwise we can trivially obtain a string outside the range by assigning opposite values to the corresponding output bits. By step $2$, we reduce our problem to an instance of $\oneinter_3$-$\avoid$ in polynomial time. Since $m>cn^2$, \cref{thm-wicket} gives a polynomial time algorithm for solving $\oneinter_3$-$\avoid$. 

This gives a polynomial time algorithm for solving $\mon$-$\NC^0_3$-$\avoid$ with quadratic stretch, thus completing the proof of the following theorem from the introduction.
\thmmonavoid*

\section{A New Turan-type Bound and Application to $\monncavoid$}

In this section, we prove a new Turan-type theorem for $3$-uniform linear hypergraphs. 
We show the application of the same to obtain an efficient algorithm for $\oneintermajthree$ with $m > n$, thus proving Theorem~\ref{thm:improvedmonavoid}. 

\subsection{Extremal Bounds for Loose $\mychi$-Cycles}

We first show an extremal bound for the loose $\mychi$ cycle in a connected $3$-uniform linear hypergraph. 

\thmextremalchi*
\begin{proof}
     Let $\H$ be a connected 3-uniform linear hypergraph with $m>n$. We would like to show that there is a loose $\mychi$ cycle in $\H$. Towards this, we construct \textit{block graph} $G$ corresponding to $\H$ as follows: let the set of $m$ hyperedges in $\H$ be the vertices of $G$ and add an edge $(f,g)$ if the hyperedges $f,g$ intersect at a vertex. Note that $G$ is a simple graph since $\H$ is linear. Observe that $G$ is connected. We have $|G|=m$ since the vertices in $G$ correspond to the hyperedges in $\H$. Now we shall argue that finding a copy of loose $\mychi$ cycle in $\H$ is equivalent to finding a subgraph $G'$ in $G$ where $G'$ consists of an edge $(f,f')$ and distinct odd walks from $g$ to $g'$ and $h$ to $h'$ without using edge $(f,f')$, where $g,h$ and $g',h'$ are neighbors of $f,f'$ respectively in $G$. The equivalence follows from the observation that a walk of length $\ell$ corresponds to a walk of length $\ell -1$ in $G$. Thus, our task is to show that $G'\subseteq G$. 
     
      For each vertex $x_i$ in $\H$, there is a clique of size $m_i$ in $G$, where $i\in [n], 1\leq m_i \leq m$. Furthermore, each vertex in $G$ participates in $3$ cliques, since $\H$ is $3$-uniform. The number of edges in $G = \sum_{i \in [n]}{m_i \choose 2}$. 
      \begin{align*}
         |E(G)|= \sum_{i \in [n]}{m_i \choose 2} &= \sum_{i \in [n]}\frac{m_i^2}{2}- \sum_{i \in [n]}\frac{m_i}{2} \\
          &= \sum_{i \in [n]}\frac{m_i^2}{2} - \frac{3m}{2} &&  \big ( \because \sum_{i \in [n]} m_i =3m \big )\\
          &\geq \frac{(3m)^2}{2n}- \frac{3m}{2}      
      \end{align*}
      The last inequality follows from the Cauchy-Schwartz inequality by taking $u$ as all $1$ vectors and $v_i$ as $m_i$. Since $m>n$, we have $|E(G)|>n$. Hence, $G$ contains a cycle. Let $e=(f,f')$ be an edge in this cycle. We would like to show that $G \setminus e$ contains $G' \setminus e$. Suffices to show that there is a walk of odd length from $g$ to $g'$. The argument for odd walk from $h$ to $h'$ is symmetric. Together with the edges $e,(f,g),(f,h),(f',g'),(f',h')$ this gives the desired subgraph $G'$. Observe that $G\setminus e$ is still connected. Hence, there must be a path $P$ from $g$ to $g'$. If this path is of odd length then we have our desired path. Otherwise, we shall obtain a walk from $g$ to $g'$ of odd length. By handshaking lemma, we have that the average degree of $G\setminus e$ is $\geq \frac{2(|E|-1)}{m}\geq \frac{2\left(\frac{(3m)^2}{2m}- \frac{3m}{2}-1\right)}{m} \geq 6-\frac{2}{m}$. By averaging argument there exists a vertex $f''$ with minimum degree $6 - \frac{2}{m}$. Since $m>2$, we have that the degree of $f''$ is at least $6$. Since $f''$ participates in at most three cliques, by PHP we get that there is a clique of size at least three which contains $f''$. Hence, there is an odd cycle $Q$ containing $f''$. Let $P'$ be a path from $g$ to $f''$, whose existence is guaranteed as $G\setminus e$ is connected. Observe that $gP'f''Qf''P'gPg'$ is a walk of odd length from $g$ to $g'$. The other walk containing $h,h'$ can be obtained similarly. Hence, there is a copy of $G'$(loose $\mychi$ cycle) in $G(\H)$. This completes the proof.   
      
\end{proof}

\subsection{Polynomial time Algorithm for $\monncavoid$ with Linear Stretch}

In section \ref{sec:wicket-proof}, we saw that using wicket as our candidate for forbidden sub-hypergraph gives a polynomial time algorithm for $\monncavoid$ with quadratic stretch. In order to improve this bound on the stretch requirement, we shall use different forbidden sub-hypergraphs which are loose $\mychi$-cycles. It suffices to show that the structure has an edge-coloring which is not $\maj$-coloring and that the extremal number for the structure is not too large.

\noindent First, we show that there exists an edge-coloring of $\mychi$ which is not $\maj$-coloring.
\begin{lemma}\label{lem:chi-cycle-not-maj-colorable}
    Let $\H$ be $\mychi$-cycle. Then there exists an edge-coloring of $\H$ which is not $\maj$-coloring. 
\end{lemma}
\begin{proof}
    Let $v_1e_1v_2\ldots v_{2\ell}e_{2\ell}v_1$ be the $\mychi$ cycle with $(e_i,e_j)$ as $\chi$-structure where $i\equiv 0 \mod 2$, $j \equiv 1 \mod 2$. Consider the following edge-coloring-
    \[
    \Gamma(e_k)= \begin{cases}
        &R, ~~~~\textrm{ if $i \equiv 0 \mod 2$}\\
        &B, ~~~~\textrm{ if $i \equiv 1 \mod 2$}
    \end{cases}
    \]
    Let $e_i \cap e_j = \{x\}$. First, consider the case when $\Pi(x)=R$. This forces the color of the other two endpoints of $e_j$ to $B$. Iteratively, we get $\Pi(v_j)= B$ iff $j \equiv 1 \mod 2$. This implies $\Gamma(e_i)=B$ which is a contradiction. The other case when $\Pi(x)=B$ is similar.
\end{proof}

\noindent For our purpose, we shall work with loose $\mychi$ cycles instead of $\mychi$ cycles. In fact, the edge-coloring demonstrated in the proof of \cref{lem:chi-cycle-not-maj-colorable} is also an edge-coloring of loose $\mychi$-cycle which is not a $\maj$-coloring.
\begin{corollary}\label{cor:loose-chi-cycle-not-maj-colorable}
    Let $\H$ be a loose $\mychi$-cycle. Then there exists an edge-coloring of $\H$ which is not a $\maj$-coloring.
\end{corollary}

\noindent Finally, combining all the results above, we shall now prove our first algorithm for $\monncavoid$ with linear stretch.

\thmimprovedmonavoid*
\begin{proof}
     By \cref{obs-cluster} and \cref{thm-2-inter}, we first reduce the circuit to a sub-circuit corresponding to cluster, and then further to the case of $\oneintermajthree$. Hence, it suffices to solve $\oneintermajthree$. By theorem \ref{thm-loose-chi-bound} we know that there exists a loose $\mychi$ in the hypergraph corresponding to the given circuit. It suffices to show that we can find loose $\mychi$ in $\H$ in polynomial time. By corollary \ref{cor:loose-chi-cycle-not-maj-colorable} and theorem \ref{lem-coloring-to-avoid} we get a polynomial time algorithm to solve $\majthreeavoid$. Indeed to find $G'$ in the proof of \cref{thm-loose-chi-bound}, we need to find the paths $P,Q'$, and a high degree vertex $f''$, which can be done in polynomial time. Hence, we can find the $G'$ in polynomial time.
\end{proof}

 \begin{algorithm}
			\begin{algorithmic}[1]
				\caption{Algorithm for solving $\maj_3$-$\avoid$ on input $C:\{0,1\}^n \to \{0,1\}^m$}\label{alg:maj3}
				\STATE {\bf Step 1:} Obtain a cluster $K$ such that $|I(K)|<|K|$ \comm{by observation \ref{obs-cluster}}
				\STATE {\bf Step 2:} Reduction from $\twointer_3$-$\avoid$ to $\oneinter_3$-$\avoid$
                \IF{$f_1,f_2 \in K$ such that $I(f_1)=\{x_1,x_2,x_3\}$ and $I(f_2)=\{x_2,x_3,x_4\}$} 
                    \STATE Set $f_1=1,f_2=0$ and $x_1=1,x_2=0,x_4=\neg x_3$, $\callo=\{f_1,f_2\}$
                    \REPEAT
                        \STATE Consider $g \in K \setminus \callo $ such that $I(g)  \neq \emptyset$
                        \IF{$| I(g) \cap I(\callo )| = 2$}
                            \STATE Let $I(g)=\{z_1,z_2,z_3\}$ and $z_3 = I(g) \setminus I(\callo )$. Let $b \in \{0,1\}$
                            \IF{Both $z_1$, $z_2$ are not alive:}
                            \STATE{If $z_1 = z_2 = b$ By setting $g = \overline{b}$, we obtain $y \not \in \range(C)$}
                            \STATE{If $z_1 = b$, $z_2 = \overline{b}$ then set $g=b$ and $y=b$}
                             \ELSIF{Exactly one of $z_1$, $z_2$ is alive: say $z_1 = b$, $z_2 = x$}  
                                \STATE Setting $g = z_3= \overline{b}$ eliminates all the variables
                            \ELSIF{Both $z_1$, $z_2$ are alive: that is $z_1=x, z_2=\lnot x$}
                                    \STATE  We set $g=z_3=0$ so that the number of variables in the cluster is at most $1$.
                            \ENDIF
                        \ELSIF{$I(g)\cap I(\callo)=\{z_1\}$}    

                            \STATE{If $z_1=b\in \{0,1\}$, then Set $g$ to $\overline{b}$ and $z_2=z_3=\overline{b}$}
                            \STATE{If $z_1=x$(alive) then set $g$ to $\overline{x}$ and $z_2=z_3=\overline{x}$}
                              
                            \STATE Update the cluster $\callo =\callo \cup \{g\}$ and $I(\callo) = I(\callo) \cup I(g)$
                        \ENDIF
                    \UNTIL{$\exists g \not \in \callo$ such that $I(g) \cap I(\callo) \neq \emptyset$}
                \ENDIF
			\STATE $\exists h_1,h_2,h_3 \not \in \callo$ such that $I(f)=\{w_1,w_2,w_3\}\subseteq I(\callo)$ since $|\callo|\le n-2$ and $|I(\callo)|=n$. 
                \IF{ $\not \exists$ a variable which is alive}
                    \STATE Setting the output of $h_1$ to $\overline{b}$ gives a string outside the range.
                \ELSIF{$\exists$ exactly one variable $w_1$ which is alive}
                    \IF{$w_2=w_3=b$} 
                        \STATE Setting the output of $h_1$ to $\overline{b}$ gives a string outside the range.
                    \ELSIF{$w_2=b,w_3=\overline{b}$}
                        \STATE Set the output of $h_1$ to $b$, which fixes $w_1$ to $b$. 
                        \STATE We have $h_2\not \in \callo$ and $I(h_2)\subseteq I(\callo)$. Same as the previous case. 
                    \ENDIF
                \ELSIF{$h_1$ has two variables that are alive and $w_3=b$}
                    \IF{$w_1=w_2=x$}
                        \STATE We set the output of $h_1$ to $\overline{b}$. This fixes all input variables (case 1).
                    \ELSIF{$w_1=x,w_2=\overline{x}$}
                        \STATE Setting the output of $h_1$ to $\overline{b}$ gives a string outside the range.
                    \ENDIF
                \ELSIF{All three variables are alive}
                    \STATE $\exists$ $h_i,h_j$ whose inputs are set the same. Setting $h_i$  to $1$ and $h_j$ to $0$ gives $y \not \in \range(C)$.
                \ENDIF
                \STATE {\bf Step 3:} Polynomial time algorithm for $\oneinter_3$-$\avoid$
                    \STATE Find a loose $\mychi$ in $H_C$ \comm{by \cref{thm-loose-chi-bound}}
                    \STATE Output $y \not \in \range(C)$ \comm{by \cref{lem-coloring-to-avoid}}
			\end{algorithmic}
		\end{algorithm}

\section{Conclusion}
We described the formulation of special cases of $\avoid$ in terms of Turan-type problems in $k$-uniform hypergraphs and demonstrated some applications to solve monotone versions of $\avoid$ under depth restrictions for the circuit. We exhibited several different fixed hypergraphs $H$ - $k$-cage, weak Fano plane, $3 \times 3$ grid, $(3,3)$-butterfly, $(3,3)$-odd kite which can be also used to derive polynomial time algorithms for $\oneintermajthree$ which in turn is used to solve $\monncavoid$ when the stretch is quadratic. Finally, the improvement to linear stretch (\cref{thm:improvedmonavoid}) comes from using loose $\mychi$ cycles as our hypergraph. We note that this is not a fixed-sized hypergraph; however, we can find it in polynomial time. Prior to this work, the best known polynomial time algorithm $\monncavoid$ was via the algorithm for $\NC^0_3$-$\avoid$ which requires $m>\frac{n^2}{\log n}$ \cite{GGNS23}. We extend our framework to solve $\oneintermajk$ for linear and quadratic stretch respectively. \cite{GGNS23} give a polynomial time algorithm for solving $\NC^0_k$-$\avoid$ when $m>\frac{n^{k-1}}{\log n}$. it would be interesting to improve the stretch using the hypergraph framework even for monotone $\NC^0_6$ circuits, which would in turn give an improvement for the case of $\NC^0_3$-$\avoid$. 

\section*{Acknowledgements}
The authors would like to thank the anonymous reviewers useful comments and for pointing out shortcomings in the presentation in the earlier version of the paper. 

\begin{appsection}{Counter Example graphs among the $(9,5)$-hypergraphs}{appsec:95-graphs}
\begin{minipage}{0.85\linewidth}

In continuation with the discussion towards the introduction, we write down the explicit example of a $(9,5)$-graph which does not lead to a contradiction as used in our arguments. 
It is easy to verify that every edge coloring is a $\maj$-coloring for this hypergraph.
This example explains why the power-bound conjecture in \cite{GL21} is not sufficient to provide us with the required improvements for the relation between $m$ and $n$ in our proposed algorithm for solving $\majthreeavoid$ and hence $\mon$-$\NC^0_3$-$\avoid$.

\end{minipage}
\begin{minipage}{0.03\linewidth}
    ~
\end{minipage}
\begin{minipage}{0.05\linewidth}
\begin{tikzpicture}[scale=0.6]
       \draw[-] (0,2) -- (-2,0) -- (-1,1);
        \draw[-] (0,2) -- (2,0) -- (1,1);
        \draw[-] (0,-2) -- (2,0) -- (1,-1);
        \draw[-] (-2,0) -- (0,-2) -- (-1,-1);
        \draw[-] (0,-2) -- (0,0) -- (0,2);
        
        \draw[fill=gray] (0,0) circle (0.1) ;
        \draw[fill=gray] (0,2) circle (0.1) ;
        \draw[fill=gray] (-2,0) circle (0.1) ;
        \draw[fill=gray] (2,0) circle (0.1) ;
        \draw[fill=gray] (0,-2) circle (0.1) ;
        
        \draw[fill=gray] (1,1) circle (0.1) ;
        \draw[fill=gray] (1,-1) circle (0.1) ;
        \draw[fill=gray] (-1,1) circle (0.1) ;
        \draw[fill=gray] (-1,-1) circle (0.1) ;
        \end{tikzpicture}
\end{minipage}

\end{appsection}

\bibliography{ref}
\bibliographystyle{alpha}

\newpage
\ifthenelse{\equal{\movetoappendix}{1}}{
        \appendix
        \section{Appendix}
        \includecollection{appendix}
} { }

\newpage

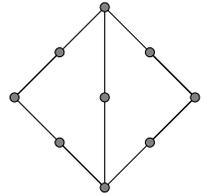
\end{document}